\documentclass[10pt,pdflatex,sn-mathphys-num]{sn-jnl}

\usepackage{thmtools}
\usepackage{graphicx}
\usepackage{multirow}
\usepackage{amsmath,amssymb,amsfonts,amsthm}
\usepackage{mathrsfs}
\usepackage[title]{appendix}
\usepackage{xcolor}
\usepackage{textcomp}
\usepackage{manyfoot}
\usepackage{booktabs}
\usepackage{algorithm}
\usepackage{algorithmicx}
\usepackage{algpseudocode}
\usepackage{listings}
\usepackage{braket}
\usepackage{dsfont}
\usepackage{ulem}
\usepackage{tabularx}
\usepackage{multirow}
\usepackage{array}
\usepackage{nicematrix}
\usepackage[position=top]{subcaption}
\usepackage{siunitx}
\usepackage{makecell}

\newcommand{\Id}{\mathds{1}}
\DeclareMathOperator{\Var}{Var}
\DeclareMathOperator{\Tr}{Tr}

\DeclareMathOperator{\EV}{\mathbb{E}}

\theoremstyle{thmstyleone}

\theoremstyle{thmstyletwo}

\theoremstyle{thmstylethree}
\newtheorem{definition}{Definition}

\theoremstyle{thmstylethree}
\newtheorem{property}{Property}

\theoremstyle{definition}

\theoremstyle{plain}

\theoremstyle{definition}

\begin{document}

\title[Private training in quantum machine learning]{Private training in quantum machine learning}

\author[1]{\fnm{Tigran} \sur{Sedrakyan}}
\author[1]{\fnm{Frédéric} \sur{Grosshans}}
\author[1,2]{\fnm{Elham} \sur{Kashefi}}

\newcommand{\todo}[1]{\textcolor{red}{#1}}
\newcommand{\rfndd}[1][ref needed]{\todo{[#1]}}

\affil[1]{\orgname{Sorbonne Université, CNRS, LIP6}, \orgaddress{\city{Paris}, \postcode{75005} \country{France}}}
\affil[2]{\orgname{School of Informatics, University of Edinburgh}, \orgaddress{\city{Edinburgh}, \postcode{EH8 9AB} \country{United Kingdom}}}

\abstract{With the emergence of machine learning (ML) models trained on large datasets containing potentially sensitive data, a major question in AI safety is how to make learning private with respect to the training data. Similar to classical machine learning, quantum machine learning (QML) models are not devoid of privacy vulnerabilities. Differential privacy (DP) is a standard tool for training ML models on sensitive data, but its impact in QML remains poorly understood. In this work we study private training in hybrid variational QML models using a classical private DP-SGD optimizer applied to pipelines with classical inputs and outputs. We analyze the interplay between gradient clipping and calibrated noise addition in DP-SGD, and its impact on optimization and accuracy for noisy and noiseless quantum models. We first explain why quantum noise does not provide a satisfactory replacement for the calibrated noise in DP-SGD for ensuring privacy. We then show how the deterministic bounds on gradient norms for a wide class of quantum models translate into explicit control of the detrimental clipping bias introduced by DP-SGD. Finally, we formulate a numerical comparison protocol under fixed clipping threshold and privacy budget and evaluate it on synthetic and image-classification tasks for equivalent quantum and classical models. Our results suggest that quantum models can retain higher accuracy in private-training regimes where the formal privacy guarantee is ensured by a classical DP-SGD mechanism.}

\keywords{differential privacy, quantum machine learning, DP-SGD}

\maketitle

\newpage
\tableofcontents
\newpage
\section{Introduction}

Learning from sensitive data while providing rigorous privacy guarantees is central to the emerging field of AI safety. Differential privacy (DP) offers an analytical way to protect individuals, bounding how much one record in a dataset can affect the output distribution of a randomized algorithm, by adding carefully calibrated noise to the outputs of a mechanism \cite{dwork2014foundations}. Over the last decade, DP has been brought into practice for training deep models through DP stochastic gradient descent (DP-SGD) \cite{abadi2016deep,mironov2017renyi,wang2019subsampled}. DP-SGD computes per-example gradients, clcalips each of them to a fixed norm threshold, and then adds calibrated Gaussian noise yielding rigorous privacy guarantees composed across training iterations, but at a nontrivial utility (accuracy) cost. In particular, per-example clipping is not merely a technical device for bounding sensitivity: it also distorts the optimization and can become a major source of performance degradation, which has been studied previously \cite{clipping-geometry,koloskova-clipbias,adaclip}.

On the other hand, a rapidly growing body of literature studies quantum machine learning (QML) for classical data \cite{qml-survey,primer-qml}, where classical data are encoded in quantum states and processed using parameterized quantum circuits. In this context, differentially private QML routines have been demonstrated, typically by combining classical DP-SGD with hybrid models, and early studies suggest that utility can be retained while enforcing a meaningful privacy budget $(\varepsilon,\delta)$ \cite{watkins2023qml,federatedqml2023,li2023qdpmeas}. The division of labor in these models is convenient analytically: it makes explicit where sensitivity arises, keeps the privacy notion identical to the one used in classical machine learning, and lets us ask whether QML enjoys a structural utility advantage at fixed privacy budget.x

In our work, we first discuss whether quantum hardware noise and shot noise can play a beneficial role in private training. We conclude that quantum noise may marginally improve empirical utility by symmetrizing the pre-clipping gradient distribution, but by itself it does not provide a practically tight replacement for standard DP-SGD in the record-level privacy setting studied here. In light of this, we focus on a simple question with practical consequences: is it possible to have an empirical advantage in privacy-utility trade-offs, by applying classical DP mechanisms to QML pipelines, compared to classical neural networks? Our starting point is that quantum machine learning (QML) is structurally different. In the class of variational quantum models considered here, the trainable quantum component evolves unitarily by design, while the model outputs are extracted through expectation values of bounded observables. Using these facts, we study privacy-preserving training on classical datasets and isolate a mechanism implying deterministic bound of DP-SGD clipping bias due to norm-controlled quantum gradients. We do not present novel privacy guarantees here; rather, using simple arguments, we support numerically the thesis that off-the-shelf classical DP mechanisms transfer naturally to hybrid QML settings, where they perform well due to structural reasons.

Similar structural reasons have been explored in classical learning literature as well. Classical models sometimes impose orthogonal or unitary parameterizations to stabilize gradients, use bounded activation functions or explicitly constrain operator norms of every layer in the network for robustness and sensitivity reasons, as observed in Lipschitz neural networks \cite{tanh-dp,arjovsky2016unitary,mhammedi2017efficient,bethune2024dpsgd}. In such models, norm preservation is an additional architectural constraint, which eliminates the need for clipping entirely, but comes with certain expressivity caveats. In the hybrid quantum setting considered here, analogous control is inherited from the model class itself through unitary evolution together with bounded readout. In numerical simulations we compare how quantum models compare to classical gradient norm-constrained Lipschitz networks.

Our work should also be distinguished from the existing privacy literature around quantum learning. Quantum differential privacy (QDP) extends privacy guarantees to quantum channels and output states \cite{hirche2022,angrisani2023unifying}. The broader privacy context around QML includes privacy guarantees in federated settings, as well as the role of hardware and measurement noise. Some proposals explore whether device-level randomness can be used for DP, for instance by calibrating to measurement noise or by composing device channels with classical noise to meet a target privacy budget \cite{li2023qdpmeas}. Some works extend this idea to private inference and adversarial robustness in QML \cite{quantum-adversarial-robustness}. Others consider formal bridges between quantum and classical DP where the privacy of a quantum output implies the privacy of all classical statistics derived from it, facilitating end-to-end reasoning about hybrid learning pipelines \cite{bridging2024}. Surveys and position papers are beginning to map the broader landscape, including risks particular to QML, evaluation gaps, and the significance of auditing~\cite{heredge2025characterizing,song2025unlearning}. Importantly, in our work we do not consider QDP applied to quantum states, but rather focus on training-time privacy of classical training data in hybrid QML pipelines, with classical inputs and outputs.

In summary, we make three contributions. First, we show that using quantum noise as the only noise source in DP-SGD is too restrictive and dilutes the privacy budget beyond practical thresholds; this claim hinges on the propagation of quantum noise to the loss gradients and a gradient bound for the class of variational quantum models considered here. Second, we show that the presence of this deterministic gradient bound allows for tighter control of the clipping bias induced by DP-SGD, providing a concrete structural explanation for why quantum models may tolerate smaller clipping thresholds than matched classical baselines under the same privacy mechanism. Finally, we formulate an experimental comparison protocol for matched quantum and classical models under identical clipping and privacy parameters, linking the theory to the observed utility in numerical study.

The rest of the paper is organized as follows. Section~\ref{sec:background} reviews the relevant background on quantum machine learning, differential privacy, and DP-SGD. Section~\ref{sec:results} fixes the learning setting and presents the main analytical results: we first analyse the propagation of quantum shot and hardware noise to loss gradients, and then derive a deterministic clipping-bias bound for quantum models. We also explain why native quantum noise is not a satisfactory replacement for the formal privacy mechanism even though it may improve utility. Section~\ref{sec:experiments} describes the numerical studies.

\section{Background}
\label{sec:background}

This section reviews the concepts used throughout the paper. We keep the discussion self-contained while emphasizing the aspects that matter most for private training: the structure of quantum machine learning models, the definition of differential privacy, and the role of clipping in DP-SGD.

\subsection{Quantum machine learning}

Quantum machine learning (QML) refers to machine learning where data are processed by parameterized quantum circuits and the resulting measurement statistics, usually Pauli expectation values, form the basis for an objective function, which is then minimized using classical gradient-based methods. For a recent and comprehensive overview of techniques and methods in QML we refer to \cite{primer-qml}. Here, we overview a common supervised learning scenario with classical data, with training dataset $D$ and entries $(x_i, y_i) \in D$, where $x_i$ are the features and $y_i$ the label.

The first step in any hybrid QML pipeline working with classical data is the encoding. Let \(\mathcal{X}\subseteq\mathbb{R}^n\) be the input domain. An encoding is a map \(\rho:\mathcal{X}\to\mathcal{H}\). A common example is amplitude encoding \(x\mapsto \rho(x) = |\phi(x)\rangle\langle\phi(x)|\) with \(|\phi(x)\rangle=\sum_{j=1}^{n} x_j |j\rangle\) for normalized $x$ with \(\|x\|_2 = 1\). Other common encoding choices include angle, basis encodings and data reuploading schemes \cite{Rath2024,PerezSalinas2020datareuploading}. Encodings control both expressivity and stability, and they determine how neighboring inputs appear in the state space: amplitude or angle encodings implement data-dependent feature maps whose induced kernels can be highly expressive, while concentrated encodings may separate classes better but also increase sensitivity.

After the encoding, a variational circuit applies a sequence of parametrized gates interleaved with fixed entanglers defining the ansatz, and a set of observables is measured to define the model output.

\begin{definition}[Hybrid variational quantum classifier]
Consider a $K$-class classification task. Let $\rho(x)$ be a data-encoding state on a Hilbert space $\mathcal{H}$,
$U(\theta)$ a parameterised unitary with $\theta\in\mathbb{R}^{p}$, and
$\{O_m\}_{m=1}^{M}$ a set of Hermitian observables on $\mathcal{H}$. The
quantum part of the model produces the expectation values
\begin{equation}
  q_m(x,\theta) \;=\; \mathrm{Tr}\!\big(O_m\, U(\theta)\,\rho(x)\,U(\theta)^{\dagger}\big),
  \qquad m = 1,\dots,M
  \label{eq:hybrid-features}
\end{equation}
A fixed or trainable classical post-processing map $h_{\phi}\colon \mathbb{R}^{M} \to \mathbb{R}^{K}$ with parameters $\phi$, turns these into class logits
\begin{equation}
  f(x,\theta,\phi) \;=\; h_{\phi}\!\big(q_{1}(x,\theta),\dots,q_{M}(x,\theta)\big),
  \label{eq:hybrid-classifier}
\end{equation}
which feed the loss $\ell(f(x,\theta,\phi),\,y)$ for a label $y\in\{1,\dots,K\}$.
\label{def:hybrid-classifier}
\end{definition}

In the simplest case $h_\phi$ is the identity map (with $M=K$) and does not depend on $\phi$, so $f(x,\theta) = \big(q_1(x,\theta),\ldots, q_K(x,\theta)\big)$. Variational classifiers then act as tunable non-classical feature maps whose outputs can be combined with standard losses and regularizers. In contrast to classical learning models, the output is not a vector of hidden activations but a set of expectation values of observables, which are optionally postprocessed classically and plugged into a classical optimisation loop.

With these ingredients, the generic supervised training setup proceeds as follows:
\begin{enumerate}
\item classical data $x_i$ are encoded as quantum states $\rho(x_i)$;
\item a variational quantum circuit $U(\theta)$ is applied;
\item expectation values of one or more observables $\{O_m\}$ are computed;
\item the measurement results are post-processed to form the objective function $f(x_i, \theta)$ and the per-example loss $\ell(f(x_i, \theta), y_i)$, with the population loss defined as $\mathcal{L}(\theta) := \EV_{(x_i,y_i)\sim D}\!\left[\ell(f(x_i, \theta), y_i)\right]$;
\item using a gradient-based (or gradient-free) optimizer, a classical subroutine updates $\theta$ in order to minimize $\mathcal{L}$
\end{enumerate}

The same loop repeats across training iterations, also known as epochs. In mini-batch stochastic gradient descent (referred to simply as SGD), at every epoch the dataset is randomly subsampled into mini-batches, each of which is processed separately, after which the loss is averaged over the entries in the mini-batch. Once all entries in the dataset have been processed once, the epoch concludes. The stochastic gradient obtained over a batch $B$, $g_B = \frac{1}{|B|} \sum_{i \in B} \partial_\theta \ell(f(x_i, \theta), y_i)$, is an unbiased estimator of the true gradient, so that $\EV[g_B] = \partial_\theta \mathcal{L}(\theta)$, meaning that SGD converges to the true optimum, provided enough training epochs.

The ansatz determines trainability and gradient concentration of the model: hardware-efficient and problem-inspired ansatzes can behave very differently with regard to barren plateaus and effective parameter counts. Highly non-local ansatzes are attractive for computational advantage in terms of expressivity, but are also vulnerable to barren plateaus, where gradients concentrate near zero as the number of qubits grows, slowing training or rendering it ineffective \cite{Larocca2025}. The primary reasons for barren plateaus are multiple: highly entangled initial states, overly expressive circuit ansatzes, global observable measurements, and noise can all contribute to their emergence. Conversely, architectures that are shallow and local are often devoid of barren plateaus and are easier to train, but may be closer to classical simulability. In this regime, gradient estimation and loss evaluation can often be performed efficiently on classical hardware, which weakens the case for a computational advantage even if learning is stable. In this work we consider quantum architectures within the classically simulable and trainable regime, so the potential for computational advantage is limited. In this scope, the quantum circuit may be executed on quantum hardware or simulated classically, but the claims apply independently of this choice.

\subsection{Differential privacy}

Differential privacy~\cite{dwork2014foundations} formalizes the idea that the presence or absence of a single individual should not significantly change the distribution of an algorithm's outputs. Informally, if two datasets are identical except for one record, then a DP mechanism should make them almost indistinguishable to any observer, regardless of the available side information. The notion is parameterized by $(\varepsilon,\delta)$, where $\varepsilon$ controls the strength of indistinguishability and $\delta$ allows for a small failure probability.

\begin{definition}[(\texorpdfstring{$\varepsilon,\delta$}{epsilon,delta})-differential privacy]\label{def:dp}
Let $\mathcal{M}$ be a randomized mechanism mapping datasets to outputs in a measurable space. For parameters $\varepsilon\ge 0$ and $\delta\in[0,1]$, $\mathcal{M}$ is $(\varepsilon,\delta)$-differentially private (DP) under a fixed adjacency relation on datasets if, for all adjacent $D \sim D'$ and all measurable $S \subseteq \mathrm{im}\,\mathcal{M}$,

\begin{equation}
\Pr[\mathcal{M}(D)\in S] \le e^\varepsilon\,\Pr[\mathcal{M}(D')\in S] + \delta
\end{equation}
\end{definition}

We will consider two datasets adjacent $D \sim D'$ (or neighboring) if they differ by exactly one record, which is present in one dataset and absent in the other, i.e. $x \notin D', D = \{x\} \cup D'$. The two prerequisites for enforcing DP are global sensitivity and noise calibration. For a vector-valued query on the dataset $Q(D)\in\mathbb{R}^d$, the $\ell_2$ global sensitivity is $\Delta_2(Q)=\sup_{D\sim D'}\|Q(D)-Q(D')\|_2$. The most commonly used Gaussian mechanism for DP releases $\mathcal{M}(D) = Q(D)+\xi$ with $\xi\sim\mathcal{N}(0,\sigma^2 I_d)$; by definition $(\varepsilon, \delta)$ corresponds to a noise scale defined by: 

\begin{equation}
\sigma_{\varepsilon, \delta} \geq \frac{\sqrt{2\ln(1.25 / \delta)}\Delta_2(Q)}{\varepsilon}
\label{eq:dp-sigma}
\end{equation}

Some important properties of DP make differentially private training of machine learning models possible. Along with the definition of DP, these properties provide a quantitative privacy guarantee that is robust to arbitrary side information and does not depend on the details of a particular privacy attack. First, DP is closed under post-processing, which formalizes that once a privatized object is released, any further, potentially randomized transformation that does not look back at the raw data preserves the guarantee. For machine learning, this is of particular importance, since gradient descent optimizers often employ some form of post-processing, such as momentum and learning-rate schedules.

\begin{property}[Post-processing]
If $\mathcal{M}$ is $(\varepsilon,\delta)$-differentially private and $\mathcal{T}$ is any (possibly randomized) map that does not access the raw dataset $D$, then $\mathcal{T}\circ \mathcal{M}$ is also $(\varepsilon,\delta)$-differentially private.
\label{prop:postprocessing}
\end{property}

Second, the composition rules ensure that when composed, private mechanisms applied to the same dataset still form a private mechanism, albeit with different privacy parameters. In iterative training approaches such as gradient descent,  this is of utmost importance and allows for estimating privacy costs across training iterations. 

\begin{property}[Composition]
If mechanisms $\mathcal{M}_1,\ldots,\mathcal{M}_T$ are applied to the same dataset and each $\mathcal{M}_t$ is $(\varepsilon_t,\delta_t)$-differentially private, then the composed mechanism is $\big(\sum_t \varepsilon_t, \sum_t \delta_t\big)$-differentially private.
\label{prop:composition}
\end{property}

 One obvious drawback of direct composition is that for a fixed $(\varepsilon, \delta)$ privacy budget of the composed mechanism, at each step $t$ one is forced to calibrate $\varepsilon_t$ to the $\delta_t$ of that step, often chosen to be $\delta_t = \delta / T$, making each $\varepsilon_t$ larger and loosening the total privacy guarantees $\varepsilon$. Sharper composition theorems exist, such as advanced composition, which allows a quadratically better scaling for $\varepsilon$ at the expense of increased $\delta$, but recent private machine learning pipelines typically use R\'enyi differential privacy \cite{mironov2017renyi} because it composes tightly and converts cleanly to $(\varepsilon,\delta)$. It is based on R\'enyi divergence defined on probability measures $P$ and $Q$:

\[
D_{\alpha}(P\|Q)\;=\;\frac{1}{\alpha-1}\log \EV_{x\sim Q}\Big[\Big(\frac{dP}{dQ}(x)\Big)^{\alpha}\Big]
\]

\noindent for order $\alpha>1$ whenever $P$ is absolutely continuous with respect to $Q$.

\begin{definition}[R\'enyi differential privacy]
A mechanism $\mathcal{M}$ satisfies $(\alpha,\varepsilon(\alpha))$-RDP if, for all adjacent datasets $D,D'$, the distributions $\mathcal{M}(D)$ and $\mathcal{M}(D')$ obey $D_{\alpha}(\mathcal{M}(D)\|\mathcal{M}(D'))\le \varepsilon(\alpha)$.
\end{definition}

In particular for the Gaussian mechanism, $\varepsilon = \frac{\alpha \Delta_t ^2}{2 \sigma_t ^ 2}$. RDP also composes additively: if $\mathcal{M}_t$ are applied on the same data, with per-step RDP curves $\varepsilon_t(\alpha)$, then the composition satisfies $\varepsilon_{\mathrm{tot}}(\alpha)=\sum_{t=1}^{T}\varepsilon_t(\alpha)$-RDP at order $\alpha$. RDP converts to $(\varepsilon,\delta)$ via

\[
\varepsilon(\delta)\;=\;\min_{\alpha>1}\Big\{\varepsilon_{\mathrm{tot}}(\alpha)\;+\;\frac{\log(1/\delta)}{\alpha-1}\Big\}
\]

Private machine learning routines harness RDP by using RDP-accountants, which maintain tight composition bounds across training epochs. They achieve this by keeping track of the curve $\varepsilon_{\mathrm{tot}}(\alpha)$, without needing to choose $\delta_t$ for each training iteration and only converting back to $(\varepsilon, \delta)$-DP for a given $\delta$ when reporting the privacy budget.

\subsection{DP-SGD and gradient clipping}

In many machine learning applications, the training set contains information about individuals, such as medical records, financial transactions, or user behavior logs, and a trained model can leak part of this information even when the raw data are never released. Such leakage can occur through model parameters, confidence scores, or repeated query access, and is commonly studied through attacks such as membership inference \cite{shokri2017membership,yeom2018privacyrisk,carlini2021lira}. In this context, privacy of the training data can be ensured at different points in the model training and deployment pipeline as shown in Figure \ref{fig:dpml-pipeline}. 

The most intuitive approach is to perturb the training data first and then rely on the DP post-processing Property~\ref{prop:postprocessing}. This can be useful in settings such as federated learning, where each client privatizes its local data before sending it to a central aggregator, which can be untrusted. This is the local-DP regime, which is stronger than the global DP setting considered here because it requires less trust in the aggregator training algorithm, but it also tends to reduce the accuracy of learned models more severely \cite{dwork2014foundations, kairouz2021fl, MahawagaArachchige2019LocalDP}. 

If training data itself was not privatized, its privacy can be ensured at the training stage or alternatively during inference. Here we focus on the former. In private training the entire training algorithm is the private mechanism, so once the model has been released, any subsequent inference query is only post-processing and does not consume additional privacy budget. By contrast, if one trains a non-private model and injects noise only at inference time, then privacy losses compose across prediction queries via Property~\ref{prop:composition}, which makes that approach poorly suited for protecting training data in publicly accessible models. Inference-time noise can still be useful for other purposes, such as adversarial robustness \cite{pixeldp,quantum-adversarial-robustness} but it is not the main object of this paper.

\begin{figure}
    \centering
    \includegraphics[width=\textwidth]{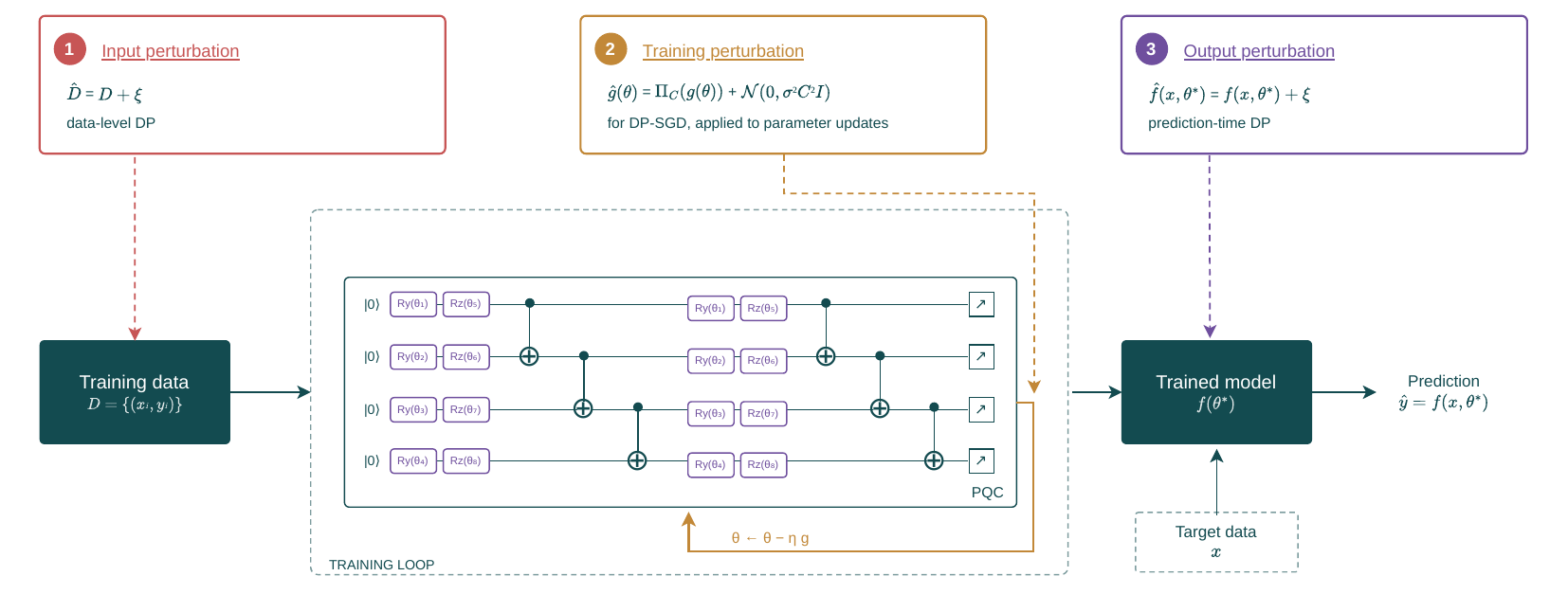}
    \caption{Schematic overview of differential privacy mechanisms in machine learning pipelines, highlighting where privatization can be applied during training and inference. Noise can be injected at input (1), training (2) or output (3)  levels. For private training, shown here is DP-SGD, applied to a training loop on a parametrized quantum circuit (PQC); however, training perturbation can happen differently and the underlying model can belong to any quantum or classical training paradigm.}
    \label{fig:dpml-pipeline}
\end{figure}

Differentially private training can be achieved by integrating randomized mechanisms into the training process so that a single entry cannot disproportionately affect the trained model. One possibility is to inject randomness directly at the model level or into the intermediate outputs of the model during training. In practice, however, because gradient computation is the last step of every optimization iteration, privacy composition is most naturally analyzed at the level of the loss gradients themselves, which also ensures the privacy of labels $y_i$, which enter the loss function directly. Privatizing the gradient updates directly is the primary principle behind differentially private stochastic gradient descent (DP-SGD) first proposed in~\cite{abadi2016deep}, which has emerged as the standard optimizer in the private deep learning literature because it provides clean and tight privacy bounds with relatively little loss of utility and does not require strong assumptions on the whole model. DP-SGD augments ordinary stochastic gradient descent in three steps. First, it computes per-example gradients so that each training record can be controlled individually. Second, it clips those gradients to a norm threshold $C$ to bound sensitivity. Third, it adds calibrated Gaussian noise to the clipped batch average. Everything that happens afterward is post-processing and therefore does not change the privacy guarantee.

The query to be privatized is the stochastic batch gradient query
\begin{equation}
Q(B)= g_B(\theta) = \frac{1}{|B|}\sum_{i\in B} g_i(\theta),
\end{equation}
where the per-example stochastic gradients 
\[
g_i(\theta):=\partial_\theta \ell(f(x_i,\theta),y_i)
\]
are an unbiased estimator of the population gradient
\begin{equation}
    g(\theta):=\EV[g_i(\theta)]=\partial_\theta \mathcal{L}(\theta),
\end{equation}
A common way to model the stochasticity of per-example gradients is to write
\begin{equation}
g_i(\theta)=g(\theta)+\xi_i(\theta),\qquad
\EV[\xi_i(\theta)]=0,
\label{eq:stoch-gradient-model}
\end{equation}
where \(\xi_i(\theta)\) denotes the centered stochastic fluctuation around the population gradient.

For the following discussion, we will use notations $i \in B$ and $(x_i, y_i) \in B$ interchangeably. For two neighboring datasets $D \sim D'$, and neighboring batches $B \sim B'$ defined so that the differing entry is present in one batch but not the other, the sensitivity of this query is
\begin{equation}
\Delta_2(Q)=\sup_{B\sim B'}\|Q(B)-Q(B')\|_2=\frac{1}{|B|}\sup_i\|g_i(\theta)\|_2
\end{equation}
Sensitivity directly affects the scale of the noise to be added in~\eqref{eq:dp-sigma} and controlling it requires a bound on the per-example gradient norm. Since most standard models do not provide such a bound naturally, one clips each per-example gradient to have norm at most $C$:
\begin{equation}
\Pi_C(g_i(\theta)):=g_i(\theta)\min\{1,C/\|g_i(\theta)\|_2\}
\label{eq:clipping}
\end{equation}
and the sensitivity is then $\Delta_2(Q) = C / |B|$. The overall DP-SGD update over a batch $B$ can be summarized as follows:
\begin{enumerate}
\item Sample a random batch $B \subset D$ and for every $i \in B$ compute the gradient $g_i(\theta)=\partial_\theta \ell(f(x_i,\theta),y_i)$.
\item Clip each gradient to the threshold $C$: $\bar g_i(\theta) = \Pi_C(g_i(\theta))$.
\item Aggregate the clipped gradients over the batch: $\bar g_B(\theta) = \frac{1}{|B|}\sum_{i\in B}\bar g_i(\theta)$.
\item Add Gaussian noise: $\widehat g_B(\theta)=\bar g_B(\theta) +Z$, where $Z\sim\mathcal N(0,\frac{\sigma^2 C^2}{|B|^2}I_p)$.
\item Update the parameters using $\widehat g_B(\theta)$ as the gradient estimate.
\end{enumerate}

This gradient descent loop then repeats for $T$ steps. One can express the noise multiplier $\sigma$ required by the full mechanism in terms of the batch subsampling rate $q=|B|/N$, the number of steps $T$, and the target privacy level $(\varepsilon,\delta)$. In particular, there exist constants $c_1,c_2>0$ such that, for any $\varepsilon<c_1 q^2T$ and $\delta > 0$, it suffices to choose
\begin{equation}
\sigma \ge c_2\,\frac{q\sqrt{T\log(1/\delta)}}{\varepsilon},
\label{eq:abadi-sigma}
\end{equation}
which yields $(\varepsilon,\delta)$-DP for the mechanism composed across $T$ steps.

The important point for optimization is that under a fixed privacy budget $(\varepsilon, \delta)$ the noise scale depends on batch size $|B|$, clipping threshold $C$ and the noise multiplier $\sigma$, which in turn depends on subsampling rate $q=|B|/N$ and number of steps $T$. Intuitively, larger batches average more per-example gradients, which reduces the variance of the stochastic gradient estimate; however, large $|B|$ also increases the sampling rate $q$, which enters the privacy accounting in \eqref{eq:abadi-sigma} and can force either fewer steps $T$ (early stopping) or a larger noise multiplier \(\sigma\). The expression for $\sigma$ is different for RDP accountants which admit sharper subsampling formulas for Poisson or fixed-size batch sampling; however, the practical dependence on the sampling rate $q$ remains \cite{mironov2017renyi,wang2019subsampled}. In practice, \(|B|\) therefore trades off per-step utility (better updates) against per-step privacy cost (faster privacy budget spending).

Likewise, clipping bounds sensitivity, but it also sets the added-noise scale. This makes the role of $C$ twofold: decreasing $C$ reduces both the sensitivity and the noise scale, but it also increases the risk of clipping and therefore of bias, inherent to clipping in the stochastic setting. To compare the impact of DP-SGD on training and gradient updates for different models under identical clipping and privacy parameters, a convenient quantity is the mean-squared error (MSE) of the stochastic update with respect to the unaltered population gradient, which captures both clipping bias and DP noise. It can be written as

\begin{equation}
\begin{aligned}
\mathrm{MSE}(\theta)&=\EV\!\big[\|\widehat g_B(\theta)-g(\theta)\|_2^2\big] \\
&=
\EV\!\big[\|\bar g_B(\theta)+Z-g(\theta)\|_2^2\big]\\
&=
\EV\!\big[\|\bar g_B(\theta)-g(\theta)\|_2^2\big]
+2\EV\!\big[\langle \bar g_B(\theta)-g(\theta),Z\rangle\big]
+\EV\|Z\|_2^2\\
&=
\EV\!\big[\|\bar g_B(\theta)-g(\theta)\|_2^2\big]
+p\Big(\frac{\sigma C}{|B|}\Big)^2\\
&=
\Tr\!\big(\mathrm{Cov}[\bar g_B(\theta)]\big)
+\|\EV[\bar g_B(\theta)]-g(\theta)\|_2^2
+p\Big(\frac{\sigma C}{|B|}\Big)^2
\end{aligned}
\label{eq:mse-bias-variance}
\end{equation}
where we used that \(Z\) is independent of the minibatch and centered, so the cross term vanishes and \(\EV\|Z\|_2^2=p(\sigma C/|B|)^2\). Under i.i.d. sampling
\[
\mathrm{Cov}(\bar g_B(\theta))=\frac{1}{|B|}\,\mathrm{Cov}(\bar g_i(\theta)),
\qquad
\EV[\bar g_B(\theta)]=\EV[\bar g_i(\theta)]
\]
Hence the same identity can be written equivalently as
\begin{equation}
\mathrm{MSE}(\theta)
=
\frac{1}{|B|}\Tr\!\big(\mathrm{Cov}[\bar g_i(\theta)]\big)
+p\Big(\frac{\sigma C}{|B|}\Big)^2
+\|\EV[\bar g_i(\theta)]-g(\theta)\|_2^2
\label{eq:mse}
\end{equation}
The first term corresponds to the sampling variance which always appears under stochastic minibatch descent. The second term is the variance of the added noise and is a result of the noise addition mechanism in DP-SGD, while the third term is the squared norm of the clipping bias
\begin{equation}
b_C(\theta):=\EV\!\big[\bar g_i(\theta)-g_i(\theta)\big] = \EV[\bar g_i(\theta)]-g(\theta)
\label{eq:def-clipping-bias}
\end{equation}
where the expectation is taken over the data distribution. The concept of clipping bias is inherent to clipped stochastic gradient descent, and it measures how much clipping shifts the expected update away from the unclipped gradient field before any privacy noise is added. Note that, unlike the other terms in \eqref{eq:mse}, the clipping bias does not improve with increasing batch size. Previous literature explored the harmful effects of clipping bias, which can accumulate into a non-vanishing optimization error floor \cite{clipping-geometry,koloskova-clipbias,adaclip}. Under clipping, the stochastic optimizer converges only to a neighborhood of a stationary point, with the size of that neighborhood controlled by the bias. But the magnitude of the bias alone is not sufficient: for optimization it also matters whether clipping preserves the descent direction. This can be measured using directional alignment
\begin{equation}
A_C(\theta)
:=
\frac{\langle \bar g_B(\theta),\,g(\theta)\rangle}{\|\bar g_B(\theta)\|_2\,\|g(\theta)\|_2}
\in[-1,1]
\label{eq:directional-alignment}
\end{equation}
When \(A_C(\theta) < 1\), stochastic clipping rotates the mean update away from the true gradient direction, which is a harmful regime for optimization. When per-example clipping is minimal, we have \(\|b_C(\theta)\|_2 \approx 0\) and \(A_C(\theta) \approx 1\). However, the directional alignment can still be favorable for optimization, \(A_C(\theta)\approx 1\), even when clipping is happening aggressively, but only causing a rescaling of the expected update. As shown in ~\cite{clipping-geometry}, the compound effect of gradient clipping on directional alignment throughout training is determined by the symmetry of the per-example gradient distribution around the true population mean, that is the distance of $\xi_i$ in \eqref{eq:stoch-gradient-model} from the closest geometrically symmetric reference distribution centered at $0$. In the perfectly symmetric case, the asymmetry-induced directional misalignment caused by clipping disappears. When $\xi_i$ is only approximately symmetric, it can still affect the training in the earlier stages of optimization, even if to a lesser extent for the impact on the convergence and, therefore, eventual accuracy. As a workaround, the authors propose adding noise to the per-example gradients before clipping, which has this desirable symmetrizing effect:
\begin{equation}
\widehat{\partial_{\theta}\ell}
=
\partial_{\theta}\ell + \mathcal N(0,v)
\label{eq:grad-symmetrification}
\end{equation}

Thus, the central question in what follows is not only how much Gaussian noise DP-SGD adds, but also how much clipping distorts the stochastic gradient before that noise is added. Boundedness of quantum model gradients can be favorable because it makes it easier to control the clipping bias.

\section{Results}
\label{sec:results}

We work in a supervised $K$-class classification setting with features $x\in\mathbb{R}^n$ and labels $y \in\{1,\dots,K\}$. In the quantum pipeline, encoding produces a state $\rho(x)=|\phi(x)\rangle\langle\phi(x)|$. For an entry $(x, y)$ a parameterized circuit
\[
U(\theta)=\prod_{j=1}^{p} e^{-i\theta_j H_j}
\]
with parameter vector $\theta\in\mathbb{R}^{p}$ acts on $\rho(x)$, and bounded observables $O_y$ are measured to define the logits
\begin{equation}
f_y(x,\theta)=\Tr\!\big(O_y\,U(\theta)\,\rho(x)\,U(\theta)^\dagger\big)
\label{eq:objective-func}
\end{equation}
For a temperature $\tau>0$, define the temperature-scaled softmax probabilities
\[
p_y(x,\theta)=\mathrm{softmax}_\tau(f_y(x,\theta)) = \frac{\exp(f_y(x,\theta)/\tau)}{\sum_{k=1}^{K}\exp(f_k(x,\theta)/\tau)},
\]
and the cross-entropy loss
\[ 
\ell(f(x,\theta), y)=-\log p_y(x,\theta)
\]

Rewriting the cross-entropy loss as
\[
\ell(f(x,\theta), y)
=
-\frac{f_y(x,\theta)}{\tau}
+\log\!\left(\sum_{r=1}^{K}e^{f_r(x,\theta)/\tau}\right)
\]
and differentiating with respect to $\theta_j$ gives
\[
\frac{\partial \ell}{\partial\theta_j}
=
\frac{1}{\tau}\left(\frac{\sum_{k=1}^{K}e^{f_k/\tau}\,\frac{\partial f_k}{\partial\theta_j}}{\sum_{r=1}^{K}e^{f_r/\tau}} - \frac{\partial f_y}{\partial\theta_j}\right)
\]
which simplifies to the closed-form loss-objective gradient relation
\begin{equation}
\frac{\partial \ell}{\partial\theta_j}
=
\frac{1}{\tau}\left(\sum_{k=1}^{K} p_k(x,\theta)\,\frac{\partial f_k(x,\theta)}{\partial\theta_j}
-
\frac{\partial f_y(x,\theta)}{\partial\theta_j}\right),
\label{eq:obj-to-loss-grad}
\end{equation}
which we will refer to in what follows.

In our quantum models noiseless gradients are computed using backpropagation, while simulation of shot noise necessitates the application of parameter-shift rules \cite{Wierichs2022generalparameter} instead. In the classical baseline we assume the same loss and dataset, a comparable number of trainable parameters, and hyperparameters tuned separately for each setting. Differential privacy is then enforced using standard DP-SGD with subsampled R\'enyi accounting.

Below we proceed with an analysis of how shot noise and hardware noise propagate to loss gradients. This reveals why such quantum noise can look superficially similar to the Gaussian perturbation used in DP-SGD, but also why that resemblance is incomplete from a privacy perspective. We then introduce the structural ingredient that DP-SGD actually needs in the hybrid setting: a deterministic clipping bias bound for the loss gradients used in training. We finally return to the original noise question and explain why native quantum noise is not a satisfactory replacement for the formal privacy mechanism.

\subsection{Noise propagation}

There are several sources of randomness in private optimization that should be distinguished. Stochastic mini-batching affects per-example gradients before clipping and, unlike the Gaussian perturbation added after averaging over a batch, does not itself provide differential privacy. In addition to the artificial DP-SGD noise, quantum implementations also involve additional sources like shot noise and hardware noise. The purpose of this subsection is to understand how those quantum sources of randomness propagate to the loss gradients, and to clarify why they should be treated as secondary to the formal DP mechanism.

The idea of using quantum noise as a resource rather than a nuisance has been explored previously. In the privacy context, recent work has shown that quantum noise can protect quantum classifiers by providing adversarial robustness and inference-stage privacy \cite{quantum-adversarial-robustness}. Those guarantees, however, are formulated for neighboring quantum inputs or query-level perturbations, not for the dataset-level adjacency relevant to training-time privacy against membership inference. Thus, it is useful to propagate the effect of quantum noise to the end of the training iteration, in order to understand how it affects training-time privacy and where  the analogy with DP-SGD breaks down.

\begin{restatable}[Quantum noise propagation to loss gradients]{theorem}{noiseprop}
\label{theorem:noise-prop}
Consider a \(K\)-class quantum classification model with logits
\[
f_k(x,\theta)=\Tr\!\big(O_k\,\rho(\theta)\big)
\]
for traceless Pauli observables \(O_k\), and cross-entropy loss \(\ell\). Under global depolarizing noise
\[
\mathcal E_\lambda(\rho)=(1-\lambda)\rho+\lambda\,\Id/d
\]
and finite-shot estimation with shot count \(S\), the induced perturbation of the loss gradient can be approximated as additive Gaussian noise:
\begin{equation}
\widehat{\partial_{\theta}\ell}
\approx
\partial_{\theta}\ell + \mathcal N(m,v),
\label{eq:gradient-plus-noise}
\end{equation}
where the mean \(m\) and variance \(v\) are computable via \ref{eq:m_and_v} from the logits, their gradients, the noise strength \(\lambda\), and the shot count \(S\).
\end{restatable}

The proof with the exact expressions for \(m\) and \(v\) can be found in Appendix \ref{app:proofs}. Theorem~\ref{theorem:noise-prop} shows that shot noise and hardware noise can, after propagation through the loss, approximately resemble an additive Gaussian perturbation on the gradients. The approximation is reasonable when \(\|\xi\|/\tau\) is not too large, which is the case for small $\lambda$ and large $S$ or $\tau$. However, this resemblance is not enough for a DP-SGD-style training guarantee.

Without an appropriate bound on $|\partial_{\theta}\ell|$, the sensitivity of gradient queries remains unbounded, necessitating clipping. In this case, note how when only shot noise is considered in the absence of depolarizing noise, we have $m = 0$ and \eqref{eq:gradient-plus-noise} is identical to \eqref{eq:grad-symmetrification}. As discussed, gradients of this form have nice geometric properties, which can reduce the clipping bias throughout the training and improve the utility under identical privacy budgets, by symmetrizing the pre-clipping gradient distribution so that the clipping step deforms the expected update less severely. Nevertheless, improvements from such a technique are expected to be marginal under reasonable assumptions on the scale of hardware noise and they can be easily simulated using a classical mechanism, as described in \cite{clipping-geometry}.

Alternatively, in order to obtain a formal privacy guarantee, one would still need a deterministic control of gradient sensitivity, that is a bound on $\partial_{\theta}\ell$ in \eqref{eq:gradient-plus-noise}. This motivates the following result.

\subsection{Clipping bias bound}

We now turn to the structural property of quantum models that makes their training under DP-SGD more favorable. Unless indicated otherwise by a subscript the norms are operator norms and we will use $(\cdot)_+ = \max\{0, \cdot\}$ for better readability.

\begin{restatable}[Deterministic clipping-bias bound]{theorem}{clipbound}
\label{theorem:grad-bound}
Consider the variational classifier of Definition~\ref{def:hybrid-classifier} with $K$ class observables $\{O_k\}_{k=1}^{K}$, parameterised unitary \(U(\theta)=\prod_{j=1}^{p}e^{-i\theta_j H_j}\), and cross-entropy loss at temperature $\tau$. Set
\[
G \;:=\; 2\,\max_{1\le k\le K}\|O_k\|\,\left(\sum_{j=1}^{p}\|H_j\|^2\right)^{1/2}
\]
Then for every data point $(x_i, y_i)$ in the dataset the per-example loss gradient $g_i(\theta)=\partial_\theta\ell(f(x_i,\theta), y_i)$ satisfies
\begin{equation}
\|g_i(\theta)\|_2 \;\le\; \frac{2G}{\tau}
\label{eq:loss-grad-bound}
\end{equation}
In particular, for every clipping threshold $C>0$, the clipping bias $b_C(\theta)$ defined in~\eqref{eq:def-clipping-bias} satisfies
\begin{equation}
\|b_C(\theta)\|_2 \;\le\; \left(\frac{2G}{\tau} - C\right)_+
\label{eq:batch-clipping-bound}
\end{equation}
\end{restatable}

Proof is postponed to Appendix \ref{app:proofs}. Theorem~\ref{theorem:grad-bound} shows that quantum models may retain more utility under DP-SGD because they admit deterministic structural control of clipping bias. The bound is deterministic, model-level, and computable directly from the circuit description. No assumption on empirical gradient tails is needed. The bound is also uniform over the encoded input: the proof uses only that the encoded datum enters as a quantum state, so the same bound holds for amplitude, angle, basis, or data-reuploading encodings, provided the trainable parameters appear only in the variational unitary layers. It is likewise independent of the gradient computation method, as long as it is exact, since this structural bound only depends on the gradient itself.

The theorem provides a direct route from the model structure to the actual DP-SGD update used in training: the object clipped by the optimizer is the per-example loss gradient $g_i(\theta)$, and its distortion is controlled explicitly by the circuit generators and observable norms. In particular, for $C = 2G/\tau$ clipping bias vanishes, so once a deterministic gradient bound is available, clipping can be tuned structurally rather than only heuristically. Smaller clipping thresholds can then be used without inducing a large clipping bias, and this directly lowers the scale \(\sigma C\) of the Gaussian noise added by DP-SGD. In practice one wants observables \(O_k\) and generators \(H_j\) to be local or low-weight, and circuits to remain sufficiently shallow, so that the norms entering \eqref{eq:main-grad-bound} do not scale unfavorably. If one instead uses highly nonlocal observables, large-support generators, or overly deep circuits, the worst-case constant can become loose enough to be uninformative for clipping calibration even though the theorem remains correct. Thus, locality is not needed for validity, but it is important for tightness and therefore for practical usefulness in private training.

Notably, this locality constraint also overlaps with the trainable regime for the models considered here. As discussed previously, models with non-local generators or observables suffer from barren plateaus and in this case the main problem is no longer the clipping bias but the vanishing gradient signal. In the trainable setting, the model also remains classically simulable, which is however not a problem by itself, as the comparison criterion is the accuracy loss under clipped (and therefore private) training, rather than expressivity or computational efficiency. This is related in spirit to Lipschitz-constrained architectures, which trade clipping bias for bias in model space by inducing gradient control as an architectural constraint built into the model class \cite{bethune2024dpsgd}. For quantum models the control is inherited from unitary evolution together with bounded expectation values.

We can now circle back to the original motivation. Theorem~\ref{theorem:noise-prop} shows that, after propagation through the loss, shot noise and hardware noise can appear as an additive Gaussian perturbation on gradient coordinates. However, this is not enough to replace the calibrated Gaussian mechanism of DP-SGD. The first obstacle is that the variance terms \(\sigma_\xi^2\), \(\sigma_{\zeta,j}^2\), and therefore the total variance \(v\), are data dependent. Differential privacy requires noise calibration that is controlled independently of the particular training example, whereas the shot-noise variance can collapse near the extreme expectation values \(\pm 1\). Introducing additional depolarizing noise can prevent the variance from vanishing at the extremes, but at the price of extra gradient bias and potentially worse trainability. The second obstacle is that in its most general form the bound in Theorem~\ref{theorem:grad-bound} can be too loose to yield competitive privacy budgets when inserted into a Gaussian-mechanism analysis. Trying to lower the budget by decreasing the number of shots also increases noise and hurts trainability or expressivity \cite{noise-in-qml,Mele2026,Larocca2025}, and can invalidate the assumption needed for Gaussian approximation.

For these reasons, using native quantum noise alone to derive record-level privacy guarantees leads to bounds that are too loose for practical training scenarios. Quantum noise can still act as an additional source of randomness and may therefore help against empirical privacy attacks, but the formal privacy guarantee should continue to come from classical DP-SGD.

\section{Numerical simulations}
\label{sec:experiments}

The setup for our numerical simulations is directly affected by the analytical insights from the previous sections. Since we are interested in comparing quantum and classical models, it is important to set up experiments that allow for a fair comparison while taking into account the metrics we have discussed. The experiments are organized to test the clipping bias explanation by asking two questions: First, at a fixed clipping threshold, do quantum models exhibit lower clipping frequency and reduced clipping bias than matched classical baselines? Second, under actual private training with the same DP-SGD mechanism, do they retain more utility?

Due to the differences between quantum and classical models, one needs to define equivalence between them to enable a fair comparison.

\begin{definition}[Equivalence between quantum and classical models (informal)]
Quantum and classical models are equivalent if they contain approximately the same number of parameters $p$ and, under the same optimizer, fine-tuned to each model separately, achieve similar accuracy on the same task.
\end{definition}

\noindent Note that this differs from most approaches in the literature in that accuracy is not a consequence of model dimension. Rather, since we are interested in the performance drop of the models when going from non-private to private settings, we design them to have equivalent performance at the same number of parameters when privacy-enhancing techniques are not employed. This practical definition of equivalence allows for comparison of models in real use-case settings.

The important nuance is that the number of parameters can affect the convergence rates so it has to stay the same between classical and quantum models. Since DP-SGD privatizes a gradient vector in \(\mathbb{R}^p\), the injected Gaussian perturbation is also \(p\)-dimensional; hence, while the per-coordinate variance is fixed by \((\sigma,C,B)\), the total noise scales as \(\EV[\|\xi_t\|_2^2]=p\,\sigma^2C^2/|B|^2\), so \(\EV[\|\xi_t\|_2]\approx (\sigma C/|B|)\sqrt{p}\), degrading the update MSE \eqref{eq:mse} as \(p\) grows. Moreover, under a fixed privacy budget \((\varepsilon,\delta)\) and sampling rate \(q=|B|/N\), composition limits the number of steps \(T\); larger models often require larger \(T\) to train, which forces either larger \(\sigma\) (more noise per step) or early stopping. The parameter count also affects DP training through clipping bias. Clipping acts on the full per-example gradient vector, and the clipping probability \(\Pr[\|g_i\|_2>C]\) typically increases with dimension at a fixed threshold \(C\), since if coordinates have comparable scale then \(\|g_i\|_2\) tends to grow like \(\sqrt{p}\). Hence larger \(p\) can lead to a higher fraction of clipped gradients and a larger bias in the direction and magnitude of the update, unless \(C\) is increased.

Fixing the number of parameters and the optimization routine, including the number of iterations, allows us to concentrate on a few tunable hyperparameters. When switching to DP training, it is important to ensure that the additional hyperparameters involved in DP-SGD are properly tuned across the models.

The clipping-only study shows that under these equivalence conditions, quantum models preserve bounded gradients and reach higher test accuracy with smaller clipping thresholds. This in turn affects the amount of necessary noise in privacy preserving simulations, allowing for better accuracy under a fixed privacy budget, and therefore a more optimal privacy-utility tradeoff. Finally, we employ private training in a quantum image loading and classification pipeline, showing how it can outperform an equivalent convolutional neural network. All the implementation details and results discussed in this section can be found in the companion Github repository \cite{sedrakyan2026privateqml}.

\subsection{Effect of Clipping}

\begin{figure}
    \centering
    \includegraphics[width=1\linewidth]{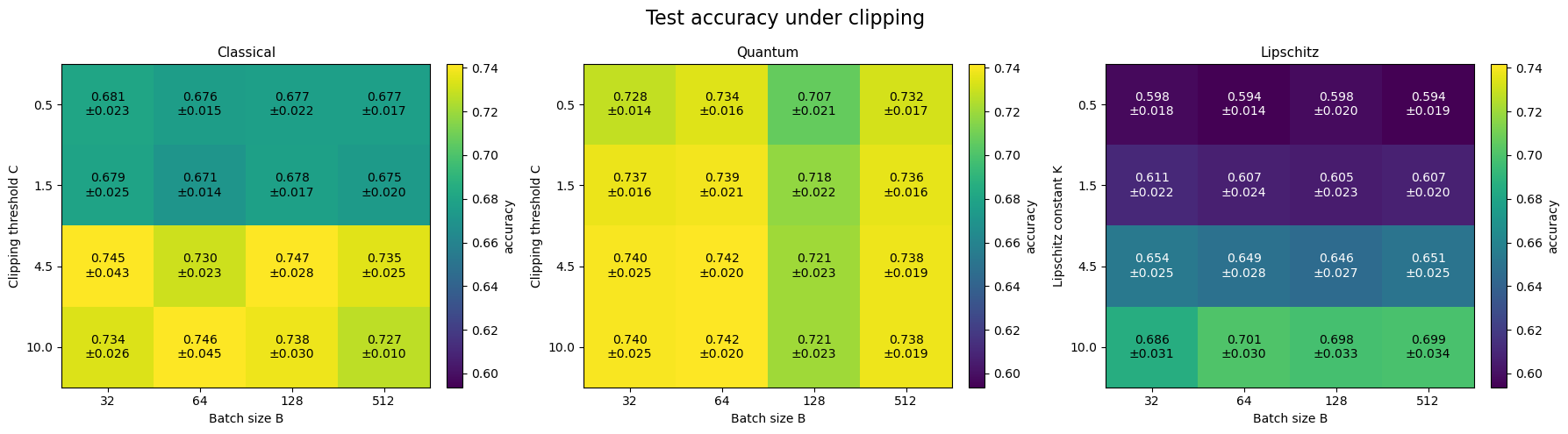}
    \caption{Final achieved test accuracy depending on the clipping threshold $C$ and the batch size $|B|$. For each configuration the learning rate is fine-tuned to ensure stable training with the highest achievable accuracy in $100$ training epochs. With appropriate tuning of the learning rate, batch size does not noticeably impact the accuracy. However, the bias floor causes the classical models with larger gradients to struggle under smaller clipping threshold $C$.}
    \label{fig:synthetic-clipped-heatmap}
\end{figure}

In order to understand the detrimental effect of clipping, we first consider experiments where privacy is not enforced. This means that the Gaussian noise in DP-SGD is not added, which allows us to isolate the effects of clipping. Such a setting has no practical meaning of its own, but it enables a cleaner numerical study of clipping effects. For the same reason we choose a synthetic tabular classification task with $N = 2048$ samples, $K = 4$ classes and $16$ features. Having strict control over the classification dataset lets us avoid outliers that could bias the results. The dataset is not perfectly separable, as some of the labels are randomly flipped to ensure non-triviality of the task.

As a classical model we take a neural network with several fully connected layers. To maximize the performance of the classifier under private training, we choose $\tanh$ functions as the activation for each layer. As discussed in \cite{tanh-dp}, the family of tempered sigmoid activations, and in particular $\tanh$, tends to outperform unbounded activations such as ReLU in private training context. The quantum classifier ansatz is built using amplitude encoding and $L$ layers of $R_Z R_Y R_Z$ followed by a CNOT ring, with Pauli-$Z$ expectation read-outs on the first $K$ qubits similar to Figure~\ref{fig:dpml-pipeline}. We also include a comparison with classical Lipschitz neural networks, structured similarly to classical neural networks, but utilizing gradient norm preserving (GNP) layers instead with Lipschitz constant matching the clipping threshold $C$. As discussed previously, Lipschitz neural networks avoid clipping bias altogether, but they induce bias in model space, limiting the expressivity of the model.

We sweep across several values of the clipping threshold $C$ and the batch size $|B|$, and for every configuration we perform $10$ experiments for $100$ training epochs with distinct random seeds on parameter initialization, reporting mean test accuracy $\pm$ one std.\ across the seeds. The results are shown in Figure~\ref{fig:synthetic-clipped-heatmap}. As can be seen, the batch size does not significantly affect training since the optimizer is tuned for each configuration separately. Moreover, the models achieve similar accuracy at a large clipping threshold of $C = 10$, which corresponds to unclipped or minimally clipped training (for Lipschitz models this corresponds to less restrictive models with more expressivity). However, the classical models clearly struggle for clipping thresholds $C \leq 1.5$, indicating detrimental effects of clipping bias.

\begin{figure}
    \centering
    \includegraphics[width=0.46\linewidth]{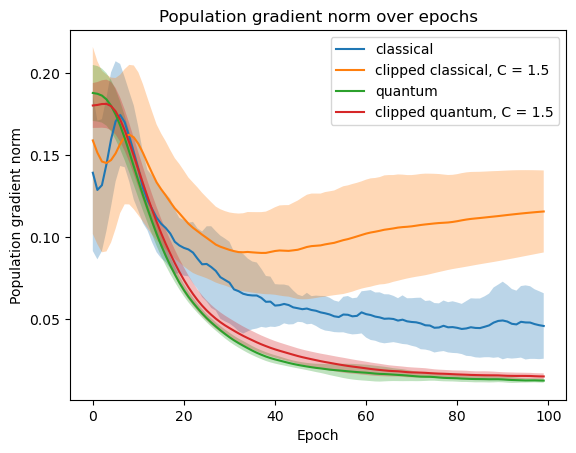}
    \includegraphics[width=0.485\linewidth]{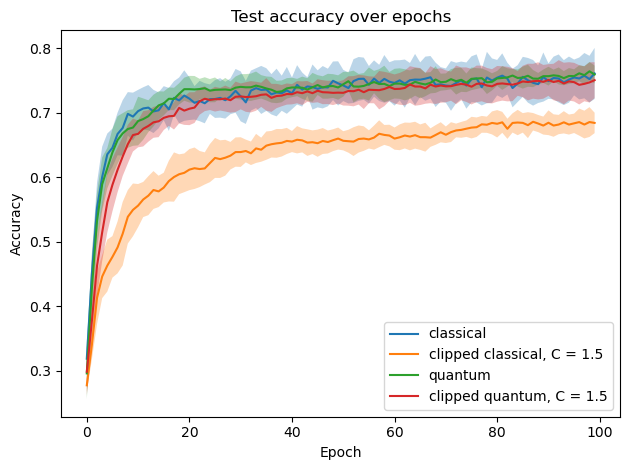}
    \caption{Evolution of population gradient and test accuracy for classifiers on synthetic data}
    \label{fig:synthetic_gradnorm_and_accuracy}
\end{figure}

\begin{figure}
    \centering
    \includegraphics[width=0.455\linewidth]{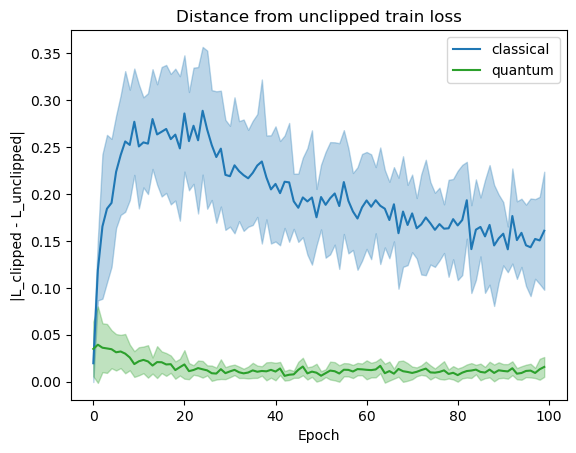}
    \includegraphics[width=0.45\linewidth]{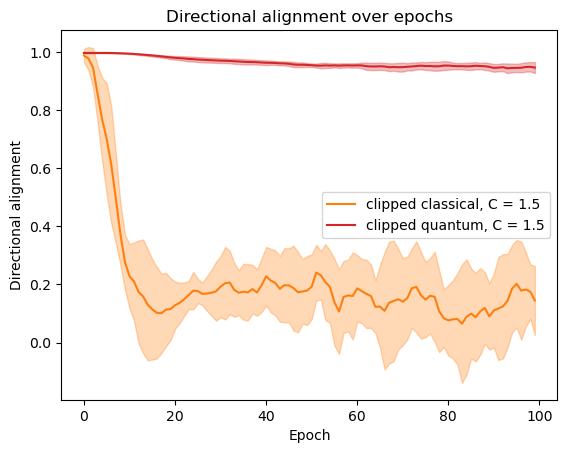}
    \caption{Evolution of distance from the unclipped loss $|\mathcal{L}_{\text{clipped}} - \mathcal{L}_{\text{unclipped}}|$ and directional alignment for classifiers on synthetic data}
    \label{fig:synthetic_deltaL_and_diralign}
\end{figure}

To better understand this phenomenon, we dive deeper into the training configuration where classical and quantum models diverge, with $C = 1.5$ and $|B| = 64$ for quantum and classical models only (since Lipschitz networks do not suffer from clipping bias). We benchmark four models: classical, quantum, clipped classical and clipped quantum. After each epoch, in addition to model accuracy and distance from the unclipped loss, we compute a diagnostic snapshot of the gradient distribution (for clipped models we use the pre-clipping distribution at the end of every epoch). Based on the complete per-example gradient distribution, we compute and report several metrics. First, we analyse the true population gradient $\|g\| = \|\partial_\theta \mathcal{L}\|$, which shrinks monotonically as expected, indicating convergence in Figure \ref{fig:synthetic_gradnorm_and_accuracy}. The only exception is the clipped classical model, for which the population gradient decreases initially but eventually plateaus. At the same time, the per-example gradient does not necessarily decrease. Looking at the mean and maximum per-example gradient norms $\EV_i\|g_i\|$ and $\max_i\|g_i\|$ in Figure \ref{fig:synthetic_mean_max_norm}, it becomes clear that a clipping threshold of $C = 1.5$ is optimal for quantum models on average and worst-case, but it is too small for classical ones, which results in a large clipping probability $p_{\text{clip}} = \Pr[\|g_i\| > C]$ (Figure \ref{fig:synthetic_biasnorm_and_clipprob}) and increased clipping bias preventing convergence to the optimum.  Importantly, clipping here significantly affects the directional alignment in Figure \ref{fig:synthetic_deltaL_and_diralign}, which prevents the optimizer from converging. Together these results indicate that quantum model performs better due to the reduced clipping bias induced by bounded gradients.

\begin{figure}
    \centering
    \includegraphics[width=0.45\linewidth]{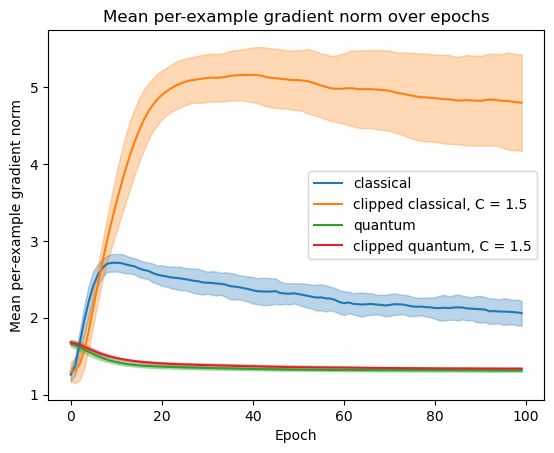}
    \includegraphics[width=0.46\linewidth]{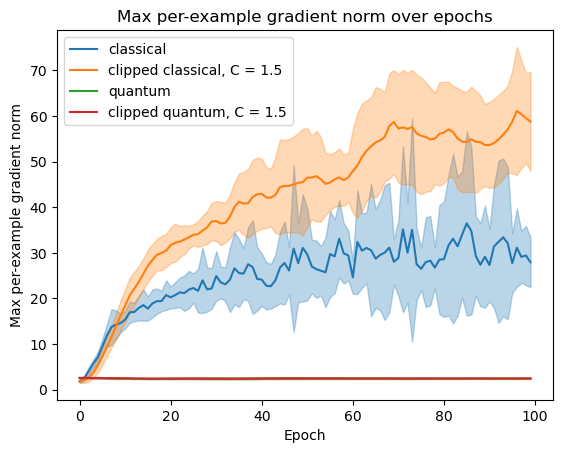}
    \caption{Evolution of mean and maximum gradient norm for classifiers on synthetic data}
    \label{fig:synthetic_mean_max_norm}
\end{figure}

\begin{figure}
    \centering
    \includegraphics[width=0.45\linewidth]{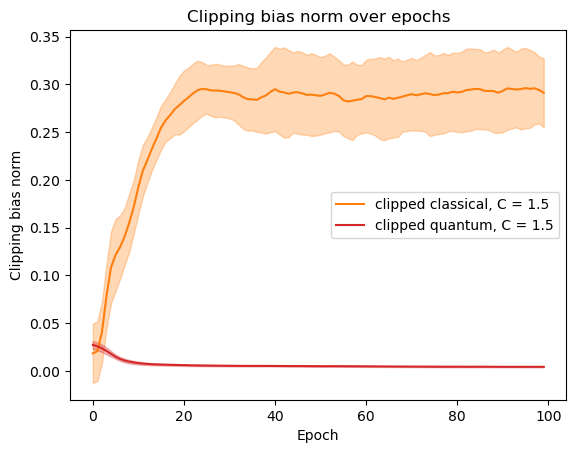}
    \includegraphics[width=0.45\linewidth]{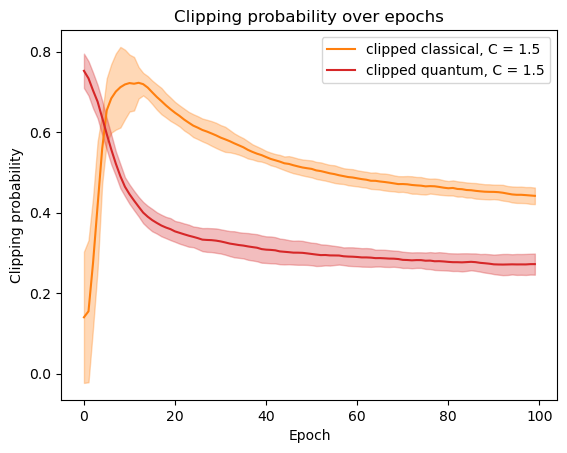}
    \caption{Evolution of clipping bias norm and clipping probability for classifiers on synthetic data}
     \label{fig:synthetic_biasnorm_and_clipprob}
\end{figure}

\subsection{Private training}

\begin{figure}
    \centering
    \includegraphics[width=1\linewidth]{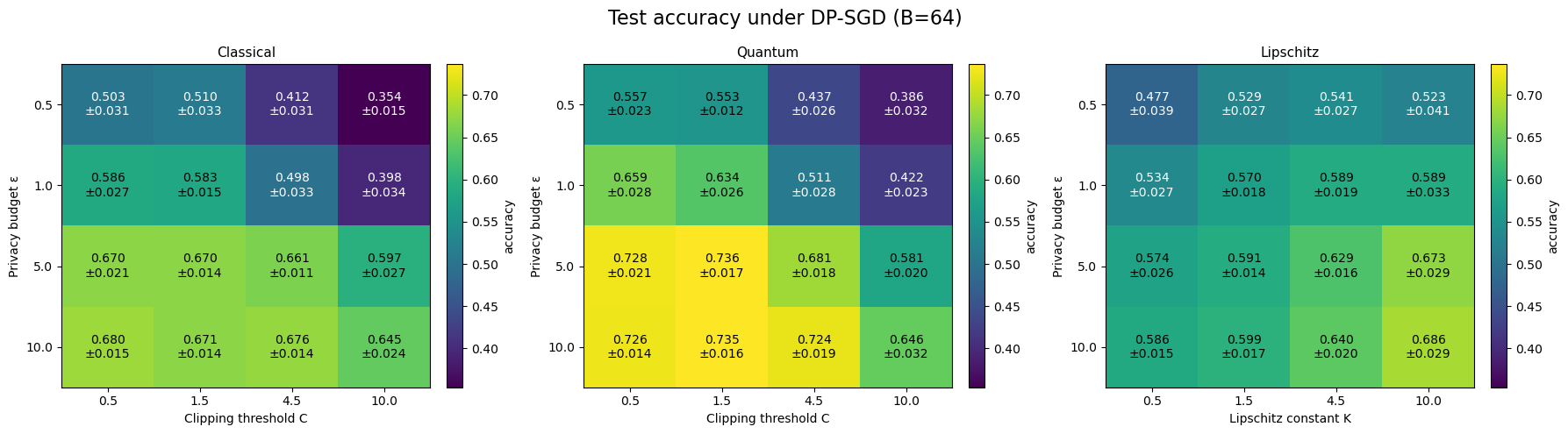}
    \caption{Final achieved test accuracy depending on the clipping threshold $C$ and the privacy budget $\varepsilon$ for batch size $|B| = 64$. For each configuration the learning rate is fine-tuned to ensure stable training with the highest achievable accuracy in $100$ training epochs. Smaller clipping thresholds reduce the added noise scale but induce clipping bias, and the impact of this double effect is visible in the figure. The optimal results are achieved when the clipping threshold tightly bounds the gradient without clipping too much, so that the scale of the added noise is reduced to a minimum. For the quantum model $C = 1.5$ ensures the best results, beating the best classical equivalent by around $5\%$ increase in accuracy. As expected, the accuracy drops regardless of the clipping threshold under tighter privacy budgets when more calibrated noise is being added. Note that the colorscheme matches that of Figure \ref{fig:synthetic-clipped-heatmap}, but the colorscale range differs to better highlight the differences.}
    \label{fig:synthetic-private-heatmap}
\end{figure}

\begin{figure}
    \centering
    \includegraphics[width=0.47\linewidth]{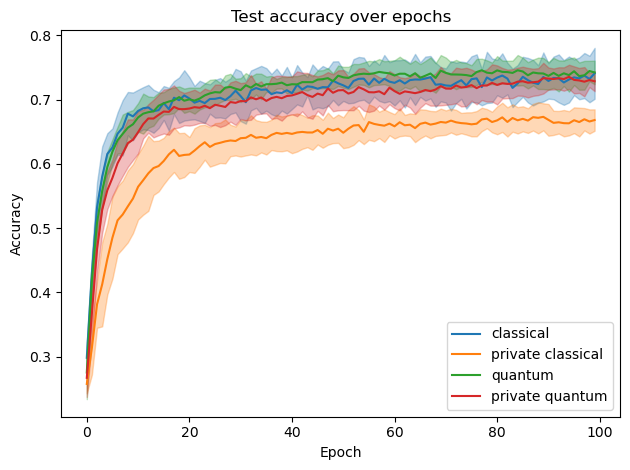}
    \includegraphics[width=0.45\linewidth]{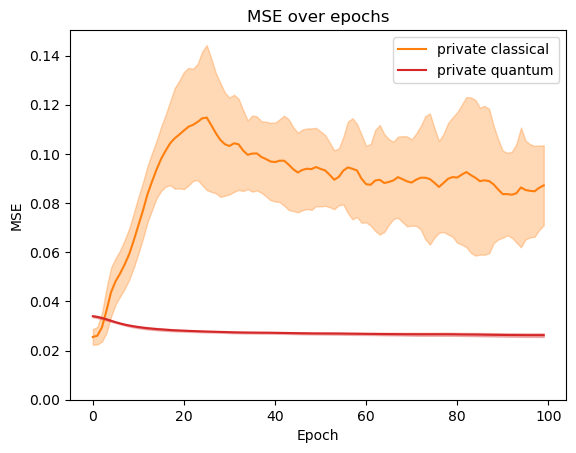}
    \caption{Evolution of accuracy and MSE for private classifiers on synthetic data with clipping threshold $C = 1.5$ and practical privacy budget of $\varepsilon = 10$  and $\delta = 10^{-5}$. MSE aggregates the effects of both clipping and calibrated noise on gradient updates. As can be seen the private classical model underperforms due to the higher MSE of gradient updates.}
    \label{fig:synthetic-private-accuracy-and-mse}
\end{figure}

Naturally, we also consider more practical settings where privacy is enforced with noise. Extending the results from the previous section, we perform a further analysis for models trained under a fixed privacy budget. Again sweeping across the clipping threshold and privacy budget values and reporting mean $\pm$ one standard deviation over 10 random initializations, we see a similar pattern here: quantum models perform better in general with smaller clipping thresholds, as shown in Figure~\ref{fig:synthetic-private-heatmap}. The reason is that small $C$ also reduces the amount of noise $\sigma C$ needed at every iteration to comply with a fixed privacy budget. We take a deeper look at this phenomenon for $C = 1.5$ and $\varepsilon = 10$ and compute the most relevant quantity that captures the effects of both clipping and Gaussian noise -- MSE~\eqref{eq:mse}. We plot its evolution throughout training alongside test accuracy in Figure~\ref{fig:synthetic-private-accuracy-and-mse}.

Finally, we apply these results to an image classification pipeline on a real dataset. To test the performance under private training on complex images, we take images from the Honda Scenes dataset \cite{narayanan2019dynamic}, which is a labeled dataset designed for dynamic scene classification. This dataset features time-based annotations detailing various aspects such as road locations, surrounding environment, weather conditions, and the state of the road surface. We sample 2000 images and classify between four classes corresponding to weather conditions and time of day: clear day, clear evening, snowy day, and snowy evening. We scale down the dimensions of images to $64\times96$ and convert them to gray-scale for simplicity. The classical network is a convolutional neural network with max-pooling and convolution layers, along with GroupNorm layers. 

Quantum image classifiers are based on a previously proposed loading and classification pipeline described in \cite{Gharibyan2026}, where there is no classical feature extraction before the quantum processing. As a first step in this pipeline quantum representations of images are learned using a technique known as hierarchical learning, which first learns coarse grained representations of images with a lower resolution, gradually scaling them up and learning more fine grained representations. Once trained, the loader parameters are fixed and loaders are then merged with the variational classifier, which outputs Pauli-$Z$ expectation values corresponding to image classes. The variational ansatz has $L$ variational layers of the same form as the one used for synthetic dataset. A schematic representation of the training pipeline is shown on Figure \ref{fig:bae-classifier}.

The equivalence conditions are ensured for both models, with each having around $500$ trainable parameters and both achieving the same accuracy under unconstrained training. As seen in Figure \ref{fig:image-classifier-results} under private training we see that the quantum classifier achieves higher accuracy and remains closer to the non-private unclipped optimization landscape, indicating the practical consequences of the bounded nature of its gradients in a private learning tasks applied to a real dataset.

\begin{figure}
    \centering
    \includegraphics[width=0.9\linewidth]{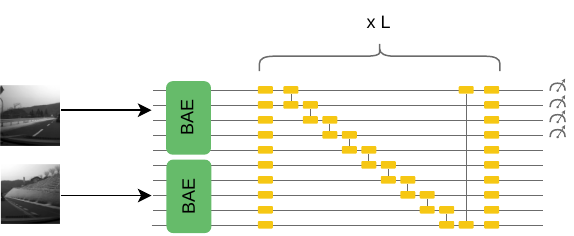}
    \caption{Quantum image loading and classification model considered here. The block amplitude encoding (BAE) image loaders are trained in advance using hierarchical training to learn representations of image blocks as quantum states. These are then plugged into the pipeline for each individual image and the $L$-layer variational classifier outputs $Z$ expectation values, which are used to train the model. Figure inspired by the ones found in \cite{Gharibyan2026}.}
    \label{fig:bae-classifier}
\end{figure}

\begin{figure}
    \centering
    \includegraphics[width=0.45\linewidth]{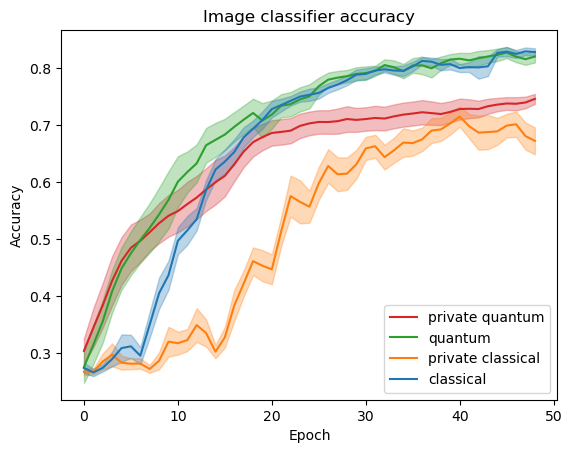} \quad
    \includegraphics[width=0.45\linewidth]{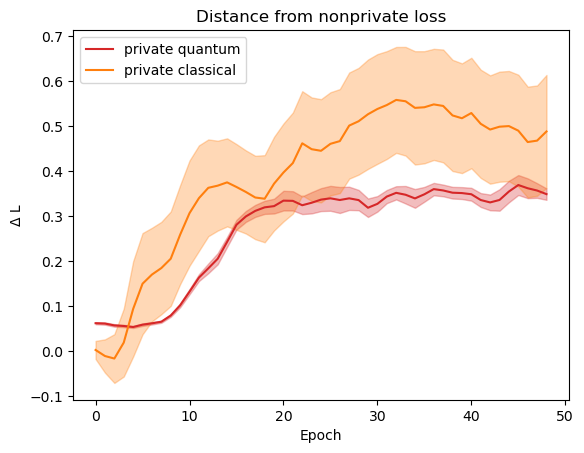}
    \caption{Evolution of accuracy and distance from the non-private loss for image classifiers under private training with $\varepsilon = 10.0$, $\delta=10^{-5}$ and $C = 1.5$. Under the clipping and noise constraints posed by the private optimizer, the quantum model achieves consistently higher classification accuracy and manages to stay closer to the unbiased optimisation landscape throughout training. The curves represent the average and standard deviation over 10 random initializations for the models.}
    \label{fig:image-classifier-results}
\end{figure}

\section{Conclusion}
\label{sec:conclusion}

In our work we analysed the propsects of privacy advantage in QML pipelines. Our noise propagation argument shows that shot noise and hardware noise can induce approximately Gaussian perturbations on loss gradients, but this resemblance is not enough for practical training-time differential privacy because the resulting variance is data dependent and the corresponding worst-case sensitivity control is too loose to serve as a competitive privacy mechanism on its own.

This observation motivated the structural result on clipping bias bound. For the class of local variational quantum models considered here, unitary evolution together with bounded observables yields deterministic control of per-example gradient norms. Through the softmax-cross-entropy link, this control transfers to the loss gradients actually used by DP-SGD, and therefore gives an explicit upper bound on clipping bias. In trainable regimes with sufficiently local observables and generators, this explains why quantum models can admit smaller effective clipping distortion than matched classical baselines under the same privacy mechanism. Moreover, as verified by our numerical results, classical Lipschitz networks with structurally imposed bounded gradients struggle in regimes where quantum models do not. At the same time, reduced clipping bias in quantum models allows them to outperform classical alternatives on synthetic and real-world tasks under private training. Our model equivalence criteria suggest that private training is a particularly natural setting for comparing quantum models with classical ones.

The resulting picture is therefore twofold. Formal record-level privacy should still come from classical DP-SGD with clipping and calibrated Gaussian noise added after aggregation. Native quantum noise is not a replacement for that mechanism; however, reduced clipping bias makes quantum models better suited to private training using DP-SGD. This also suggests a broader quantum-inspired design principle to be explored further in future work: if one designs classical architectures with analogous norm-preserving parameterizations and bounded activation/output layers, then part of the same clipping-bias advantage transfers beyond quantum hardware.

While the quantum models considered here are completely classically simulable, rapid advances in quantum hardware and the approaching era of fault tolerant quantum computing (FTQC) will allow quantum-native variational learning algorithms to benefit from the privacy results presented here. In this regime local quantum variational models will stay trainable due to the absence of noise, while still retaining the desirable boundedness property. 

\section{Statements and Declarations}

\textbf{Funding:} We acknowledge funding from Agence Nationale de la Recherche (Plan France 2030 through the project OQuLus ANR-22-PETQ-0013)

\noindent \textbf{Conflicts of interest:} We declare no conflicts of interest related to this work.

\noindent \textbf{Data Availability:} Data and code used to construct plots and tables in this work can be obtained from the companion Gihub repository \cite{sedrakyan2026privateqml}.

\noindent \textbf{Author Contributions:} T.S. initiated the project, contributed to the theory and simulation results, wrote the simulation and data aggregation code, and edited the manuscript. F.G. and E.K. contributed to the idea, edited the manuscript, and supervised the work.

\newpage

\bibliography{refs}


\begin{thebibliography}{37}
\ifx \bisbn   \undefined \def \bisbn  #1{ISBN #1}\fi
\ifx \binits  \undefined \def \binits#1{#1}\fi
\ifx \bauthor  \undefined \def \bauthor#1{#1}\fi
\ifx \batitle  \undefined \def \batitle#1{#1}\fi
\ifx \bjtitle  \undefined \def \bjtitle#1{#1}\fi
\ifx \bvolume  \undefined \def \bvolume#1{\textbf{#1}}\fi
\ifx \byear  \undefined \def \byear#1{#1}\fi
\ifx \bissue  \undefined \def \bissue#1{#1}\fi
\ifx \bfpage  \undefined \def \bfpage#1{#1}\fi
\ifx \blpage  \undefined \def \blpage #1{#1}\fi
\ifx \burl  \undefined \def \burl#1{\textsf{#1}}\fi
\ifx \doiurl  \undefined \def \doiurl#1{\url{https://doi.org/#1}}\fi
\ifx \betal  \undefined \def \betal{\textit{et al.}}\fi
\ifx \binstitute  \undefined \def \binstitute#1{#1}\fi
\ifx \binstitutionaled  \undefined \def \binstitutionaled#1{#1}\fi
\ifx \bctitle  \undefined \def \bctitle#1{#1}\fi
\ifx \beditor  \undefined \def \beditor#1{#1}\fi
\ifx \bpublisher  \undefined \def \bpublisher#1{#1}\fi
\ifx \bbtitle  \undefined \def \bbtitle#1{#1}\fi
\ifx \bedition  \undefined \def \bedition#1{#1}\fi
\ifx \bseriesno  \undefined \def \bseriesno#1{#1}\fi
\ifx \blocation  \undefined \def \blocation#1{#1}\fi
\ifx \bsertitle  \undefined \def \bsertitle#1{#1}\fi
\ifx \bsnm \undefined \def \bsnm#1{#1}\fi
\ifx \bsuffix \undefined \def \bsuffix#1{#1}\fi
\ifx \bparticle \undefined \def \bparticle#1{#1}\fi
\ifx \barticle \undefined \def \barticle#1{#1}\fi
\bibcommenthead
\ifx \bconfdate \undefined \def \bconfdate #1{#1}\fi
\ifx \botherref \undefined \def \botherref #1{#1}\fi
\ifx \url \undefined \def \url#1{\textsf{#1}}\fi
\ifx \bchapter \undefined \def \bchapter#1{#1}\fi
\ifx \bbook \undefined \def \bbook#1{#1}\fi
\ifx \bcomment \undefined \def \bcomment#1{#1}\fi
\ifx \oauthor \undefined \def \oauthor#1{#1}\fi
\ifx \citeauthoryear \undefined \def \citeauthoryear#1{#1}\fi
\ifx \endbibitem  \undefined \def \endbibitem {}\fi
\ifx \bconflocation  \undefined \def \bconflocation#1{#1}\fi
\ifx \arxivurl  \undefined \def \arxivurl#1{\textsf{#1}}\fi
\csname PreBibitemsHook\endcsname

\bibitem[\protect\citeauthoryear{Dwork and Roth}{2014}]{dwork2014foundations}
\begin{barticle}
\bauthor{\bsnm{Dwork}, \binits{C.}},
\bauthor{\bsnm{Roth}, \binits{A.}}:
\batitle{The algorithmic foundations of differential privacy}.
\bjtitle{Found. Trends Theor. Comput. Sci.}
\bvolume{9}(\bissue{3–4}),
\bfpage{211}--\blpage{407}
(\byear{2014})
\doiurl{10.1561/0400000042}
\end{barticle}
\endbibitem

\bibitem[\protect\citeauthoryear{Abadi et~al.}{2016}]{abadi2016deep}
\begin{bchapter}
\bauthor{\bsnm{Abadi}, \binits{M.}},
\bauthor{\bsnm{Chu}, \binits{A.}},
\bauthor{\bsnm{Goodfellow}, \binits{I.}},
\bauthor{\bsnm{McMahan}, \binits{H.B.}},
\bauthor{\bsnm{Mironov}, \binits{I.}},
\bauthor{\bsnm{Talwar}, \binits{K.}},
\bauthor{\bsnm{Zhang}, \binits{L.}}:
\bctitle{Deep learning with differential privacy}.
In: \bbtitle{Proceedings of the 2016 ACM SIGSAC Conference on Computer and Communications Security},
pp. \bfpage{308}--\blpage{318}
(\byear{2016}).
\doiurl{10.1145/2976749.2978318}
\end{bchapter}
\endbibitem

\bibitem[\protect\citeauthoryear{Mironov}{2017}]{mironov2017renyi}
\begin{botherref}
\oauthor{\bsnm{Mironov}, \binits{I.}}:
R{\'e}nyi differential privacy.
2017 IEEE 30th Computer Security Foundations Symposium (CSF),
263--275
(2017)
\doiurl{10.48550/arXiv.1702.07476}
{\href{https://arxiv.org/abs/1702.07476}{{arXiv:1702.07476}}}
\end{botherref}
\endbibitem

\bibitem[\protect\citeauthoryear{Wang et~al.}{2021}]{wang2019subsampled}
\begin{botherref}
\oauthor{\bsnm{Wang}, \binits{Y.-X.}},
\oauthor{\bsnm{Balle}, \binits{B.}},
\oauthor{\bsnm{Kasiviswanathan}, \binits{S.}}:
Subsampled rényi differential privacy and analytical moments accountant.
Journal of Privacy and Confidentiality
\textbf{10}(2)
(2021)
\doiurl{10.29012/jpc.723}
\end{botherref}
\endbibitem

\bibitem[\protect\citeauthoryear{Chen et~al.}{2020}]{clipping-geometry}
\begin{bchapter}
\bauthor{\bsnm{Chen}, \binits{X.}},
\bauthor{\bsnm{Wu}, \binits{Z.S.}},
\bauthor{\bsnm{Hong}, \binits{M.}}:
\bctitle{Understanding gradient clipping in private sgd: a geometric perspective}.
In: \bbtitle{Proceedings of the 34th International Conference on Neural Information Processing Systems}.
\bsertitle{NIPS '20}.
\bpublisher{Curran Associates Inc.},
\blocation{Red Hook, NY, USA}
(\byear{2020})
\end{bchapter}
\endbibitem

\bibitem[\protect\citeauthoryear{Koloskova et~al.}{2023}]{koloskova-clipbias}
\begin{bchapter}
\bauthor{\bsnm{Koloskova}, \binits{A.}},
\bauthor{\bsnm{Hendrikx}, \binits{H.}},
\bauthor{\bsnm{Stich}, \binits{S.U.}}:
\bctitle{Revisiting gradient clipping: Stochastic bias and tight convergence guarantees}.
In: \beditor{\bsnm{Krause}, \binits{A.}},
\beditor{\bsnm{Brunskill}, \binits{E.}},
\beditor{\bsnm{Cho}, \binits{K.}},
\beditor{\bsnm{Engelhardt}, \binits{B.}},
\beditor{\bsnm{Sabato}, \binits{S.}},
\beditor{\bsnm{Scarlett}, \binits{J.}} (eds.)
\bbtitle{Proceedings of the 40th International Conference on Machine Learning}.
\bsertitle{Proceedings of Machine Learning Research},
vol. \bseriesno{202},
pp. \bfpage{17343}--\blpage{17363}.
\bpublisher{PMLR},
\blocation{Honolulu, Hawaii, USA}
(\byear{2023}).
\burl{https://proceedings.mlr.press/v202/koloskova23a.html}
\end{bchapter}
\endbibitem

\bibitem[\protect\citeauthoryear{Pichapati et~al.}{2019}]{adaclip}
\begin{botherref}
\oauthor{\bsnm{Pichapati}, \binits{V.}},
\oauthor{\bsnm{Suresh}, \binits{A.T.}},
\oauthor{\bsnm{Yu}, \binits{F.X.}},
\oauthor{\bsnm{Reddi}, \binits{S.J.}},
\oauthor{\bsnm{Kumar}, \binits{S.}}:
AdaCliP: Adaptive Clipping for Private SGD
(2019).
\url{https://arxiv.org/abs/1908.07643}
\end{botherref}
\endbibitem

\bibitem[\protect\citeauthoryear{Zaman et~al.}{2023}]{qml-survey}
\begin{botherref}
\oauthor{\bsnm{Zaman}, \binits{K.}},
\oauthor{\bsnm{Marchisio}, \binits{A.}},
\oauthor{\bsnm{Hanif}, \binits{M.A.}},
\oauthor{\bsnm{Shafique}, \binits{M.}}:
{A Survey on Quantum Machine Learning: Current Trends, Challenges, Opportunities, and the Road Ahead}
(2023)
{\href{https://arxiv.org/abs/2310.10315}{{arXiv:2310.10315}}}
\end{botherref}
\endbibitem

\bibitem[\protect\citeauthoryear{Chang and Cerezo}{2025}]{primer-qml}
\begin{botherref}
\oauthor{\bsnm{Chang}, \binits{S.Y.}},
\oauthor{\bsnm{Cerezo}, \binits{M.}}:
{A Primer on Quantum Machine Learning}
(2025)
{\href{https://arxiv.org/abs/2511.15969}{{arXiv:2511.15969}}}
\end{botherref}
\endbibitem

\bibitem[\protect\citeauthoryear{Watkins et~al.}{2023}]{watkins2023qml}
\begin{barticle}
\bauthor{\bsnm{Watkins}, \binits{W.M.}},
\bauthor{\bsnm{Chen}, \binits{S.Y.-C.}},
\bauthor{\bsnm{Yoo}, \binits{S.}}:
\batitle{Quantum machine learning with differential privacy}.
\bjtitle{Scientific Reports}
\bvolume{13},
\bfpage{2023}
(\byear{2023})
\doiurl{10.1038/s41598-022-24082-z}
{\href{https://arxiv.org/abs/2103.06232}{{arXiv:2103.06232}}}
\end{barticle}
\endbibitem

\bibitem[\protect\citeauthoryear{Rofougaran et~al.}{2024}]{federatedqml2023}
\begin{bchapter}
\bauthor{\bsnm{Rofougaran}, \binits{R.}},
\bauthor{\bsnm{Yoo}, \binits{S.}},
\bauthor{\bsnm{Tseng}, \binits{H.-H.}},
\bauthor{\bsnm{Chen}, \binits{S.Y.-C.}}:
\bctitle{Federated quantum machine learning with differential privacy}.
In: \bbtitle{ICASSP 2024 - 2024 IEEE International Conference on Acoustics, Speech and Signal Processing (ICASSP)},
pp. \bfpage{9811}--\blpage{9815}
(\byear{2024}).
\doiurl{10.1109/ICASSP48485.2024.10447155}
\end{bchapter}
\endbibitem

\bibitem[\protect\citeauthoryear{Li et~al.}{2023}]{li2023qdpmeas}
\begin{botherref}
\oauthor{\bsnm{Li}, \binits{Y.}},
\oauthor{\bsnm{Zhao}, \binits{Y.}},
\oauthor{\bsnm{Zhang}, \binits{X.}},
\oauthor{\bsnm{Zhong}, \binits{H.}},
\oauthor{\bsnm{Pan}, \binits{M.}},
\oauthor{\bsnm{Zhang}, \binits{C.}}:
Differential privacy preserving quantum computing via measurement noise
(2023)
{\href{https://arxiv.org/abs/2312.08210}{{arXiv:2312.08210}}}
\end{botherref}
\endbibitem

\bibitem[\protect\citeauthoryear{Papernot et~al.}{2021}]{tanh-dp}
\begin{barticle}
\bauthor{\bsnm{Papernot}, \binits{N.}},
\bauthor{\bsnm{Thakurta}, \binits{A.}},
\bauthor{\bsnm{Song}, \binits{S.}},
\bauthor{\bsnm{Chien}, \binits{S.}},
\bauthor{\bsnm{Erlingsson}, \binits{{\'U}.}}:
\batitle{Tempered sigmoid activations for deep learning with differential privacy}.
\bjtitle{Proceedings of the AAAI Conference on Artificial Intelligence}
\bvolume{35}(\bissue{10}),
\bfpage{9312}--\blpage{9321}
(\byear{2021})
\doiurl{10.1609/aaai.v35i10.17123}
\end{barticle}
\endbibitem

\bibitem[\protect\citeauthoryear{Arjovsky et~al.}{2016}]{arjovsky2016unitary}
\begin{bchapter}
\bauthor{\bsnm{Arjovsky}, \binits{M.}},
\bauthor{\bsnm{Shah}, \binits{A.}},
\bauthor{\bsnm{Bengio}, \binits{Y.}}:
\bctitle{Unitary evolution recurrent neural networks}.
In: \beditor{\bsnm{Balcan}, \binits{M.F.}},
\beditor{\bsnm{Weinberger}, \binits{K.Q.}} (eds.)
\bbtitle{Proceedings of The 33rd International Conference on Machine Learning}.
\bsertitle{Proceedings of Machine Learning Research},
vol. \bseriesno{48},
pp. \bfpage{1120}--\blpage{1128}.
\bpublisher{PMLR},
\blocation{New York, New York, USA}
(\byear{2016}).
\burl{https://proceedings.mlr.press/v48/arjovsky16.html}
\end{bchapter}
\endbibitem

\bibitem[\protect\citeauthoryear{Mhammedi et~al.}{2017}]{mhammedi2017efficient}
\begin{bchapter}
\bauthor{\bsnm{Mhammedi}, \binits{Z.}},
\bauthor{\bsnm{Hellicar}, \binits{A.}},
\bauthor{\bsnm{Rahman}, \binits{A.}},
\bauthor{\bsnm{Bailey}, \binits{J.}}:
\bctitle{Efficient orthogonal parametrisation of recurrent neural networks using householder reflections}.
In: \beditor{\bsnm{Precup}, \binits{D.}},
\beditor{\bsnm{Teh}, \binits{Y.W.}} (eds.)
\bbtitle{Proceedings of the 34th International Conference on Machine Learning}.
\bsertitle{Proceedings of Machine Learning Research},
vol. \bseriesno{70},
pp. \bfpage{2401}--\blpage{2409}.
\bpublisher{PMLR},
\blocation{Sydney, Australia}
(\byear{2017}).
\burl{https://proceedings.mlr.press/v70/mhammedi17a.html}
\end{bchapter}
\endbibitem

\bibitem[\protect\citeauthoryear{B{\'{e}}thune et~al.}{2024}]{bethune2024dpsgd}
\begin{bchapter}
\bauthor{\bsnm{B{\'{e}}thune}, \binits{L.}},
\bauthor{\bsnm{Massena}, \binits{T.}},
\bauthor{\bsnm{Boissin}, \binits{T.}},
\bauthor{\bsnm{Bellet}, \binits{A.}},
\bauthor{\bsnm{Mamalet}, \binits{F.}},
\bauthor{\bsnm{Prudent}, \binits{Y.}},
\bauthor{\bsnm{Friedrich}, \binits{C.}},
\bauthor{\bsnm{Serrurier}, \binits{M.}},
\bauthor{\bsnm{Vigouroux}, \binits{D.}}:
\bctitle{{DP-SGD} without clipping: The lipschitz neural network way}.
In: \bbtitle{The Twelfth International Conference on Learning Representations}
(\byear{2024}).
\burl{https://openreview.net/forum?id=BEyEziZ4R6}
\end{bchapter}
\endbibitem

\bibitem[\protect\citeauthoryear{Hirche et~al.}{2023}]{hirche2022}
\begin{barticle}
\bauthor{\bsnm{Hirche}, \binits{C.}},
\bauthor{\bsnm{Rouzé}, \binits{C.}},
\bauthor{\bsnm{França}, \binits{D.S.}}:
\batitle{Quantum differential privacy: An information theory perspective}.
\bjtitle{IEEE Transactions on Information Theory}
\bvolume{69}(\bissue{9}),
\bfpage{5771}--\blpage{5787}
(\byear{2023})
\doiurl{10.1109/TIT.2023.3272904}
\end{barticle}
\endbibitem

\bibitem[\protect\citeauthoryear{Angrisani et~al.}{2023}]{angrisani2023unifying}
\begin{botherref}
\oauthor{\bsnm{Angrisani}, \binits{A.}},
\oauthor{\bsnm{Doosti}, \binits{M.}},
\oauthor{\bsnm{Kashefi}, \binits{E.}}:
A unifying framework for differentially private quantum algorithms
(2023).
\url{https://arxiv.org/abs/2307.04733}
\end{botherref}
\endbibitem

\bibitem[\protect\citeauthoryear{Du et~al.}{2021}]{quantum-adversarial-robustness}
\begin{barticle}
\bauthor{\bsnm{Du}, \binits{Y.}},
\bauthor{\bsnm{Hsieh}, \binits{M.-H.}},
\bauthor{\bsnm{Liu}, \binits{T.}},
\bauthor{\bsnm{Tao}, \binits{D.}},
\bauthor{\bsnm{Liu}, \binits{N.}}:
\batitle{Quantum noise protects quantum classifiers against adversaries}.
\bjtitle{Phys. Rev. Res.}
\bvolume{3},
\bfpage{023153}
(\byear{2021})
\doiurl{10.1103/PhysRevResearch.3.023153}
\end{barticle}
\endbibitem

\bibitem[\protect\citeauthoryear{Zhao et~al.}{2024}]{bridging2024}
\begin{bchapter}
\bauthor{\bsnm{Zhao}, \binits{Y.}},
\bauthor{\bsnm{Zhong}, \binits{H.}},
\bauthor{\bsnm{Zhang}, \binits{X.}},
\bauthor{\bsnm{Li}, \binits{Y.}},
\bauthor{\bsnm{Zhang}, \binits{C.}},
\bauthor{\bsnm{Pan}, \binits{M.}}:
\bctitle{Bridging quantum computing and differential privacy: Insights into quantum computing privacy}.
In: \bbtitle{2024 IEEE International Conference on Quantum Computing and Engineering (QCE)},
vol. \bseriesno{01},
pp. \bfpage{13}--\blpage{24}
(\byear{2024}).
\doiurl{10.1109/QCE60285.2024.00012}
\end{bchapter}
\endbibitem

\bibitem[\protect\citeauthoryear{Heredge et~al.}{2025}]{heredge2025characterizing}
\begin{botherref}
\oauthor{\bsnm{Heredge}, \binits{J.}},
\oauthor{\bsnm{Kumar}, \binits{N.}},
\oauthor{\bsnm{Herman}, \binits{D.}},
\oauthor{\bsnm{Chakrabarti}, \binits{S.}},
\oauthor{\bsnm{Yalovetzky}, \binits{R.}},
\oauthor{\bsnm{Sureshbabu}, \binits{S.H.}},
\oauthor{\bsnm{Li}, \binits{C.}},
\oauthor{\bsnm{Pistoia}, \binits{M.}}:
Characterizing privacy in quantum machine learning.
npj Quantum Information
\textbf{11}
(2025)
\doiurl{10.1038/s41534-025-01022-z}
\end{botherref}
\endbibitem

\bibitem[\protect\citeauthoryear{Su et~al.}{2025}]{song2025unlearning}
\begin{botherref}
\oauthor{\bsnm{Su}, \binits{J.}},
\oauthor{\bsnm{He}, \binits{R.}},
\oauthor{\bsnm{Li}, \binits{G.}},
\oauthor{\bsnm{Qin}, \binits{S.}},
\oauthor{\bsnm{He}, \binits{Z.}},
\oauthor{\bsnm{Situ}, \binits{H.}},
\oauthor{\bsnm{Gao}, \binits{F.}}:
From membership-privacy leakage to quantum machine unlearning
(2025)
{\href{https://arxiv.org/abs/2509.06086}{{arXiv:2509.06086}}}
\end{botherref}
\endbibitem

\bibitem[\protect\citeauthoryear{Rath and Date}{2024}]{Rath2024}
\begin{barticle}
\bauthor{\bsnm{Rath}, \binits{M.}},
\bauthor{\bsnm{Date}, \binits{H.}}:
\batitle{Quantum data encoding: a comparative analysis of classical-to-quantum mapping techniques and their impact on machine learning accuracy}.
\bjtitle{EPJ Quantum Technology}
\bvolume{11}(\bissue{1}),
\bfpage{72}
(\byear{2024})
\doiurl{10.1140/epjqt/s40507-024-00285-3}
\end{barticle}
\endbibitem

\bibitem[\protect\citeauthoryear{P{\'{e}}rez-Salinas et~al.}{2020}]{PerezSalinas2020datareuploading}
\begin{barticle}
\bauthor{\bsnm{P{\'{e}}rez-Salinas}, \binits{A.}},
\bauthor{\bsnm{Cervera-Lierta}, \binits{A.}},
\bauthor{\bsnm{Gil-Fuster}, \binits{E.}},
\bauthor{\bsnm{Latorre}, \binits{J.I.}}:
\batitle{Data re-uploading for a universal quantum classifier}.
\bjtitle{{Quantum}}
\bvolume{4},
\bfpage{226}
(\byear{2020})
\doiurl{10.22331/q-2020-02-06-226}
\end{barticle}
\endbibitem

\bibitem[\protect\citeauthoryear{Larocca et~al.}{2025}]{Larocca2025}
\begin{barticle}
\bauthor{\bsnm{Larocca}, \binits{M.}},
\bauthor{\bsnm{Thanasilp}, \binits{S.}},
\bauthor{\bsnm{Wang}, \binits{S.}},
\bauthor{\bsnm{Sharma}, \binits{K.}},
\bauthor{\bsnm{Biamonte}, \binits{J.}},
\bauthor{\bsnm{Coles}, \binits{P.J.}},
\bauthor{\bsnm{Cincio}, \binits{L.}},
\bauthor{\bsnm{McClean}, \binits{J.R.}},
\bauthor{\bsnm{Holmes}, \binits{Z.}},
\bauthor{\bsnm{Cerezo}, \binits{M.}}:
\batitle{Barren plateaus in variational quantum computing}.
\bjtitle{Nature Reviews Physics}
\bvolume{7}(\bissue{4}),
\bfpage{174}--\blpage{189}
(\byear{2025})
\doiurl{10.1038/s42254-025-00813-9}
\end{barticle}
\endbibitem

\bibitem[\protect\citeauthoryear{Shokri et~al.}{2017}]{shokri2017membership}
\begin{bchapter}
\bauthor{\bsnm{Shokri}, \binits{R.}},
\bauthor{\bsnm{Stronati}, \binits{M.}},
\bauthor{\bsnm{Song}, \binits{C.}},
\bauthor{\bsnm{Shmatikov}, \binits{V.}}:
\bctitle{Membership inference attacks against machine learning models}.
In: \bbtitle{IEEE Symposium on Security and Privacy},
pp. \bfpage{3}--\blpage{18}
(\byear{2017}).
\doiurl{10.1109/SP.2017.41}
\end{bchapter}
\endbibitem

\bibitem[\protect\citeauthoryear{Yeom et~al.}{2018}]{yeom2018privacyrisk}
\begin{botherref}
\oauthor{\bsnm{Yeom}, \binits{S.}},
\oauthor{\bsnm{Giacomelli}, \binits{I.}},
\oauthor{\bsnm{Fredrikson}, \binits{M.}},
\oauthor{\bsnm{Jha}, \binits{S.}}:
Privacy risk in machine learning: Analyzing the connection to overfitting.
2018 IEEE Computer Security Foundations Symposium (CSF),
268--282
(2018)
\doiurl{10.1109/CSF.2018.00027}
{\href{https://arxiv.org/abs/1709.01604}{{arXiv:1709.01604}}}
\end{botherref}
\endbibitem

\bibitem[\protect\citeauthoryear{Carlini et~al.}{2022}]{carlini2021lira}
\begin{bchapter}
\bauthor{\bsnm{Carlini}, \binits{N.}},
\bauthor{\bsnm{Chien}, \binits{S.}},
\bauthor{\bsnm{Nasr}, \binits{M.}},
\bauthor{\bsnm{Song}, \binits{S.}},
\bauthor{\bsnm{Terzis}, \binits{A.}},
\bauthor{\bsnm{Tramèr}, \binits{F.}}:
\bctitle{Membership inference attacks from first principles}.
In: \bbtitle{2022 IEEE Symposium on Security and Privacy (SP)},
pp. \bfpage{1897}--\blpage{1914}
(\byear{2022}).
\doiurl{10.1109/SP46214.2022.9833649}
\end{bchapter}
\endbibitem

\bibitem[\protect\citeauthoryear{Kairouz et~al.}{2021}]{kairouz2021fl}
\begin{barticle}
\bauthor{\bsnm{Kairouz}, \binits{P.}},
\bauthor{\bsnm{McMahan}, \binits{H.B.}},
\bauthor{\bsnm{Avent}, \binits{B.}},
\bauthor{\bsnm{Bellet}, \binits{A.}},
\bauthor{\bsnm{Bennis}, \binits{M.}},
\bauthor{\bsnm{Bhagoji}, \binits{A.N.}},
\bauthor{\bsnm{Bonawitz}, \binits{K.}},
\bauthor{\bsnm{Charles}, \binits{Z.}},
\bauthor{\bsnm{Cormode}, \binits{G.}},
\bauthor{\bsnm{Cummings}, \binits{R.}},
\bauthor{\bsnm{D'Oliveira}, \binits{R.G.L.}},
\bauthor{\bsnm{Eichner}, \binits{H.}},
\bauthor{\bsnm{El~Rouayheb}, \binits{S.}},
\bauthor{\bsnm{Evans}, \binits{D.}},
\bauthor{\bsnm{Gardner}, \binits{J.}},
\bauthor{\bsnm{Garrett}, \binits{Z.}},
\bauthor{\bsnm{Gasc{\'o}n}, \binits{A.}},
\bauthor{\bsnm{Ghazi}, \binits{B.}},
\bauthor{\bsnm{Gibbons}, \binits{P.B.}},
\bauthor{\bsnm{Gruteser}, \binits{M.}},
\bauthor{\bsnm{Harchaoui}, \binits{Z.}},
\bauthor{\bsnm{He}, \binits{C.}},
\bauthor{\bsnm{He}, \binits{L.}},
\bauthor{\bsnm{Huo}, \binits{Z.}},
\bauthor{\bsnm{Hutchinson}, \binits{B.}},
\bauthor{\bsnm{Hsu}, \binits{J.}},
\bauthor{\bsnm{Jaggi}, \binits{M.}},
\bauthor{\bsnm{Javidi}, \binits{T.}},
\bauthor{\bsnm{Joshi}, \binits{G.}},
\bauthor{\bsnm{Khodak}, \binits{M.}},
\bauthor{\bsnm{Kone{\v{c}}n{\'y}}, \binits{J.}},
\bauthor{\bsnm{Korolova}, \binits{A.}},
\bauthor{\bsnm{Koushanfar}, \binits{F.}},
\bauthor{\bsnm{Koyejo}, \binits{S.}},
\bauthor{\bsnm{Lepoint}, \binits{T.}},
\bauthor{\bsnm{Liu}, \binits{Y.}},
\bauthor{\bsnm{Mittal}, \binits{P.}},
\bauthor{\bsnm{Mohri}, \binits{M.}},
\bauthor{\bsnm{Nock}, \binits{R.}},
\bauthor{\bsnm{{\"O}zg{\"u}r}, \binits{A.}},
\bauthor{\bsnm{Pagh}, \binits{R.}},
\bauthor{\bsnm{Qi}, \binits{H.}},
\bauthor{\bsnm{Ramage}, \binits{D.}},
\bauthor{\bsnm{Raskar}, \binits{R.}},
\bauthor{\bsnm{Song}, \binits{D.}},
\bauthor{\bsnm{Song}, \binits{W.}},
\bauthor{\bsnm{Stich}, \binits{S.U.}},
\bauthor{\bsnm{Sun}, \binits{Z.}},
\bauthor{\bsnm{Suresh}, \binits{A.T.}},
\bauthor{\bsnm{Tram{\`e}r}, \binits{F.}},
\bauthor{\bsnm{Vepakomma}, \binits{P.}},
\bauthor{\bsnm{Wang}, \binits{J.}},
\bauthor{\bsnm{Xiong}, \binits{L.}},
\bauthor{\bsnm{Xu}, \binits{Z.}},
\bauthor{\bsnm{Yang}, \binits{Q.}},
\bauthor{\bsnm{Yu}, \binits{F.X.}},
\bauthor{\bsnm{Yu}, \binits{H.}},
\bauthor{\bsnm{Zhao}, \binits{S.}}:
\batitle{Advances and open problems in federated learning}.
\bjtitle{Foundations and Trends in Machine Learning}
\bvolume{14}(\bissue{1--2}),
\bfpage{1}--\blpage{210}
(\byear{2021})
\doiurl{10.1561/2200000083}
\end{barticle}
\endbibitem

\bibitem[\protect\citeauthoryear{Mahawaga~Arachchige et~al.}{2020}]{MahawagaArachchige2019LocalDP}
\begin{barticle}
\bauthor{\bsnm{Mahawaga~Arachchige}, \binits{P.C.}},
\bauthor{\bsnm{Bertok}, \binits{P.}},
\bauthor{\bsnm{Khalil}, \binits{I.}},
\bauthor{\bsnm{Liu}, \binits{D.}},
\bauthor{\bsnm{Camtepe}, \binits{S.}},
\bauthor{\bsnm{Atiquzzaman}, \binits{M.}}:
\batitle{Local differential privacy for deep learning}.
\bjtitle{IEEE Internet of Things Journal}
\bvolume{7}(\bissue{7}),
\bfpage{5827}--\blpage{5842}
(\byear{2020})
\doiurl{10.1109/JIOT.2019.2952146}
\end{barticle}
\endbibitem

\bibitem[\protect\citeauthoryear{Lecuyer et~al.}{2019}]{pixeldp}
\begin{bchapter}
\bauthor{\bsnm{Lecuyer}, \binits{M.}},
\bauthor{\bsnm{Atlidakis}, \binits{V.}},
\bauthor{\bsnm{Geambasu}, \binits{R.}},
\bauthor{\bsnm{Hsu}, \binits{D.}},
\bauthor{\bsnm{Jana}, \binits{S.}}:
\bctitle{Certified robustness to adversarial examples with differential privacy}.
In: \bbtitle{2019 IEEE Symposium on Security and Privacy (SP)},
pp. \bfpage{656}--\blpage{672}
(\byear{2019}).
\doiurl{10.1109/SP.2019.00044}
\end{bchapter}
\endbibitem

\bibitem[\protect\citeauthoryear{Wierichs et~al.}{2022}]{Wierichs2022generalparameter}
\begin{barticle}
\bauthor{\bsnm{Wierichs}, \binits{D.}},
\bauthor{\bsnm{Izaac}, \binits{J.}},
\bauthor{\bsnm{Wang}, \binits{C.}},
\bauthor{\bsnm{Lin}, \binits{C.Y.-Y.}}:
\batitle{General parameter-shift rules for quantum gradients}.
\bjtitle{{Quantum}}
\bvolume{6},
\bfpage{677}
(\byear{2022})
\doiurl{10.22331/q-2022-03-30-677}
\end{barticle}
\endbibitem

\bibitem[\protect\citeauthoryear{Khanal and Rivas}{2023}]{noise-in-qml}
\begin{bchapter}
\bauthor{\bsnm{Khanal}, \binits{B.}},
\bauthor{\bsnm{Rivas}, \binits{P.}}:
\bctitle{Evaluating the impact of noise on variational quantum circuits in nisq era devices}.
In: \bbtitle{2023 Congress in Computer Science, Computer Engineering, \& Applied Computing (CSCE)},
pp. \bfpage{1658}--\blpage{1664}
(\byear{2023}).
\doiurl{10.1109/CSCE60160.2023.00272}
\end{bchapter}
\endbibitem

\bibitem[\protect\citeauthoryear{Mele et~al.}{2026}]{Mele2026}
\begin{barticle}
\bauthor{\bsnm{Mele}, \binits{A.A.}},
\bauthor{\bsnm{Angrisani}, \binits{A.}},
\bauthor{\bsnm{Ghosh}, \binits{S.}},
\bauthor{\bsnm{Khatri}, \binits{S.}},
\bauthor{\bsnm{Eisert}, \binits{J.}},
\bauthor{\bsnm{Stilck~Fran{\c{c}}a}, \binits{D.}},
\bauthor{\bsnm{Quek}, \binits{Y.}}:
\batitle{Noise-induced shallow circuits and the absence of barren plateaus}.
\bjtitle{Nature Physics}
\bvolume{22}(\bissue{5}),
\bfpage{751}--\blpage{756}
(\byear{2026})
\doiurl{10.1038/s41567-026-03245-z}
\end{barticle}
\endbibitem

\bibitem[\protect\citeauthoryear{Sedrakyan}{2026}]{sedrakyan2026privateqml}
\begin{botherref}
\oauthor{\bsnm{Sedrakyan}, \binits{T.}}:
private-qml: Code for "Private training in quantum machine learning".
GitHub
(2026).
\url{https://github.com/tigran-sedrakyan/private-qml}
\end{botherref}
\endbibitem

\bibitem[\protect\citeauthoryear{Narayanan et~al.}{2019}]{narayanan2019dynamic}
\begin{botherref}
\oauthor{\bsnm{Narayanan}, \binits{A.}},
\oauthor{\bsnm{Dwivedi}, \binits{I.}},
\oauthor{\bsnm{Dariush}, \binits{B.}}:
Dynamic Traffic Scene Classification with Space-Time Coherence
(2019).
\url{https://arxiv.org/abs/1905.12708}
\end{botherref}
\endbibitem

\bibitem[\protect\citeauthoryear{Gharibyan et~al.}{2026}]{Gharibyan2026}
\begin{barticle}
\bauthor{\bsnm{Gharibyan}, \binits{H.}},
\bauthor{\bsnm{Karapetyan}, \binits{H.}},
\bauthor{\bsnm{Sedrakyan}, \binits{T.}},
\bauthor{\bsnm{Subasic}, \binits{P.}},
\bauthor{\bsnm{Su}, \binits{V.P.}},
\bauthor{\bsnm{Tanin}, \binits{R.H.}},
\bauthor{\bsnm{Tepanyan}, \binits{H.}}:
\batitle{Quantum image loading and classification: experiments on utility-scale quantum computers}.
\bjtitle{Quantum Machine Intelligence}
\bvolume{8}(\bissue{1}),
\bfpage{57}
(\byear{2026})
\doiurl{10.1007/s42484-026-00388-3}
\end{barticle}
\endbibitem

\end{thebibliography}

\newpage

\appendix
\section{Proofs of Theorems}
\label{app:proofs}

\noiseprop*
\begin{proof}
First, for simplicity, consider a model outputting an expectation value of a single observable. Let \(f(\theta)=\Tr\!\big(O\,\rho(\theta)\big)\) be an ideal expectation-value objective and \(\rho(\theta)=U(\theta)\rho(x) U(\theta)^\dagger\), where we omit the dependence of \(f\) on \(x\). For a fixed parameter coordinate \(\theta_j\), consider the standard two-shift parameter-shift gradient without shot noise:
\[
g_j(\theta)=\partial_{\theta_j}f(\theta)
=
\frac{1}{2}\Big(\mu_j^{+}(\theta)-\mu_j^{-}(\theta)\Big),
\]
where
\[
\mu_j^{\pm}(\theta):=f\!\left(\theta\pm \frac{\pi}{2}e_j\right)
\]
are the shifted expectations and \(e_j\) is the \(j\)-th standard basis vector. Let the corresponding finite-shot estimators be
\[
\widehat{\mu}_j^{\pm}(\theta):=\widehat f\!\left(\theta\pm \frac{\pi}{2}e_j\right),
\qquad
\widehat g_j(\theta)
=
\frac{1}{2}\Big(\widehat{\mu}_j^{+}(\theta)-\widehat{\mu}_j^{-}(\theta)\Big)
\]

Since each $O_k$ is a traceless Pauli operator, its single-shot eigenvalue outcomes take values in $\{\pm 1\}$. Let \(X_{\pm,1},\dots,X_{\pm,S}\in\{\pm 1\}\) denote the single-shot outcomes for \(S\) shots. Then
\[
\widehat{\mu}_j^{\pm}(\theta)=\frac{1}{S}\sum_{s=1}^{S}X_{\pm,s},
\qquad
\EV[X_{\pm,s}]=\mu_j^{\pm}(\theta),
\]
and therefore
\[
\Var\!\big[\widehat{\mu}_j^{\pm}(\theta)\big]
=
\frac{1}{S}\Var[X_{\pm,s}]
=
\frac{1}{S}\Big(1-(\mu_j^{\pm}(\theta))^2\Big)
\]
In particular,
\[
\Var[\widehat f(\theta)]
=
\frac{1}{S}\Big(1-f(\theta)^2\Big)
\]
Writing \(\widehat f(\theta)=f(\theta)+\xi(\theta)\), where \(\xi(\theta)\) is the zero-mean shot-noise term, gives
\[
\Var[\xi(\theta)]
=
\frac{1}{S}\Big(1-f(\theta)^2\Big)
\]
Similarly, writing \(\widehat g_j(\theta)=g_j(\theta)+\zeta_j(\theta)\), and because the two shifted shot sets are independent, we obtain
\[
\Var[\zeta_j(\theta)]
=
\frac{1}{4}\Big(\Var[\widehat{\mu}_j^{+}(\theta)]+\Var[\widehat{\mu}_j^{-}(\theta)]\Big)
=
\frac{1}{4S}\Big(2-\big((\mu_j^{+}(\theta))^2+(\mu_j^{-}(\theta))^2\big)\Big)
\]

To model hardware noise, consider a CPTP channel \(\mathcal{E}_\lambda\) applied to the ideal state before measurement, with noise strength \(\lambda\). The noisy objective is
\[
f_{\lambda}(\theta)=\Tr\!\big(O\,\mathcal E_\lambda(\rho(\theta))\big),
\]
and the corresponding infinite-shot gradient coordinate can be written as
\[
g_{j,\lambda}(\theta)=\partial_{\theta_j} f_{\lambda}(\theta)=g_j(\theta)+b_{j,\lambda}(\theta),
\]
where \(b_{j,\lambda}(\theta)\) is a generally nonzero bias term. For the global depolarizing channel
\[
\mathcal E_\lambda(\rho)=(1-\lambda)\rho+\lambda\,\Id/d,
\]
one has
\[
f_{\lambda}(\theta)
=
\Tr\!\big(O\,\mathcal E_\lambda(\rho(\theta))\big)
=
(1-\lambda)\Tr\!\big(O\,\rho(\theta)\big)+\lambda\,\Tr\!\Big(O\,\frac{\Id}{d}\Big)
=
(1-\lambda)\,f(\theta)
\]
for traceless Pauli observables, and therefore
\[
g_{j,\lambda}(\theta)=(1-\lambda)\,g_j(\theta)
\]
So the objective and the gradient coordinate are multiplicatively shrunk by the factor \(1-\lambda\).

If shot noise is also present, then the shifted noisy expectations satisfy
\[
\mu_{j,\lambda}^{\pm}(\theta):=f_{\lambda}\!\left(\theta\pm \frac{\pi}{2}e_j\right),
\qquad
\widehat{\mu}_{j,\lambda}^{\pm}(\theta):=\widehat f_{\lambda}\!\left(\theta\pm \frac{\pi}{2}e_j\right),
\]
and the corresponding variance becomes
\[
\Var\!\big[\widehat{\mu}^{\pm}_{j,\lambda}(\theta)\big]
=
\frac{1}{S}\Big(1-(\mu_{j,\lambda}^{\pm}(\theta))^2\Big)
\]
In particular,
\[
\Var[\widehat f_{\lambda}(\theta)]
=
\frac{1}{S}\Big(1-f_{\lambda}(\theta)^2\Big)
=
\frac{1}{S}\Big(1-(1-\lambda)^2f(\theta)^2\Big)
\]
Writing
\[
\widehat f_{\lambda}(\theta)=f(\theta)+\xi_{\lambda}(\theta),
\qquad
\widehat g_{j,\lambda}(\theta)=g_j(\theta)+\zeta_{j,\lambda}(\theta),
\]
gives
\[
\sigma_\xi^2 := \Var[\xi_{\lambda}(\theta)]
=
\frac{1}{S}\Big(1-(1-\lambda)^2f(\theta)^2\Big),
\]
\[
\sigma_{\zeta,j}^2 := \Var[\zeta_{j,\lambda}(\theta)]
=
\frac{1}{4S}\Big(2-(1-\lambda)^2\big((\mu_j^{+}(\theta))^2+(\mu_j^{-}(\theta))^2\big)\Big)
\]

In contrast to the shot-noise-only case, these perturbations are not centered around zero relative to the ideal model. Indeed,
\[
\mu^\xi := \EV[\xi_{\lambda}(\theta)] = -\lambda f(\theta),
\qquad
\mu^\zeta_j := \EV[\zeta_{j,\lambda}(\theta)] = -\lambda g_j(\theta)
\]
For large \(S\), the central limit theorem implies that the scalar perturbations are approximately Gaussian. We now apply this componentwise to the $K$-class logits. Write
\[
f := (f_1(x,\theta),\dots,f_K(x,\theta))^\top,
\qquad
\widehat f := (\widehat f_1(x,\theta),\dots,\widehat f_K(x,\theta))^\top,
\qquad
\xi := \widehat f-f,
\]
For a fixed parameter index \(j\), let
\[
g^{(j)} := (g_{j,1},\dots,g_{j,K})^\top,
\qquad
\widehat g^{(j)} := (\widehat g_{j,1},\dots,\widehat g_{j,K})^\top,
\qquad
\zeta^{(j)} := \widehat g^{(j)}-g^{(j)},
\]
Let
\[
\qquad
p_y(f):=\mathrm{softmax}_\tau(f)=  \mathrm{sm}(f) = \frac{e^{f_y/\tau}}{\sum_k e^{f_k / \tau}},
\qquad
\widehat p_y(f):=\mathrm{sm}(\widehat f) = \mathrm{sm}(f+\xi)
\]
Then the noisy estimator of the loss-gradient coordinate is
\[
\widehat{\partial_{\theta_j}\ell}
=
\frac{1}{\tau}\Bigg(
\sum_{k=1}^{K}\widehat p_k\Big(g_{j,k}(x,\theta)+\zeta_{j,k}(x,\theta)\Big)
-
\Big(g_{j,y}(x,\theta)+\zeta_{j,y}(x,\theta)\Big)
\Bigg)
\]
Setting \(\xi=0\) and \(\zeta^{(j)}=0\) recovers the noiseless gradient~\eqref{eq:obj-to-loss-grad}:
\[
\partial_{\theta_j}\ell
=
\frac{1}{\tau}\Bigg(\sum_{k=1}^{K}p_k\,g_{j,k}(x,\theta)-g_{j,y}(x,\theta)\Bigg)
\]

Since \(\widehat p=\mathrm{sm}(f+\xi)\) is nonlinear in \(\xi\), we use a first-order delta-method approximation and drop second-order products of the noise terms. Write
\[
\frac{\partial p_k}{\partial f_r}=\frac{1}{\tau}p_k(\delta_{kr}-p_r),
\qquad
J_p(f):=\partial_f\,\mathrm{sm}(f)=\frac{1}{\tau}(\mathrm{diag}(p)-pp^\top)
\]
Hence, linearizing the softmax around \(f\) gives
\[
\widehat p \approx p + J_p(f)\xi,
\qquad
\delta p := \widehat p-p \approx J_p(f)\xi
\]
Substituting \(\widehat p=p+\delta p\) into the partial derivative estimator gives
\[
\widehat{\partial_{\theta_j}\ell}
=
\frac{1}{\tau}\Big(
\sum_{k=1}^{K} p_k\big(g_{j,k}+\zeta_{j,k}\big)
+
\sum_{k=1}^{K}\delta p_k\big(g_{j,k}+\zeta_{j,k}\big)
-
\big(g_{j,y}+\zeta_{j,y}\big)
\Big)
\]
Keeping only first-order terms in \(\xi\) and \(\zeta^{(j)}\) yields
\[
\widehat{\partial_{\theta_j}\ell}
\approx
\frac{1}{\tau}\Big(\sum_{k=1}^{K} p_k g_{j,k}-g_{j,y}\Big)
+
\frac{1}{\tau}\Big(\sum_{k=1}^{K}p_k\zeta_{j,k}-\zeta_{j,y}\Big)
+
\frac{1}{\tau}\sum_{k=1}^{K}\delta p_k\,g_{j,k}
\]
Replacing \(\delta p = J_p(f)\xi\) and defining
\[
w^{(j)} := \tau J_p(f)g^{(j)}
\]
gives
\[
\widehat{\partial_{\theta_j}\ell}
\approx
\partial_{\theta_j}\ell
+
\underbrace{\frac{1}{\tau}\Big(p^\top\zeta^{(j)}-\zeta_{j,y}\Big)}_{\eta_j^{(\zeta)}}
+
\underbrace{\frac{1}{\tau^2}(w^{(j)})^\top\xi}_{\eta_j^{(\xi)}}
\]
Writing
\[
a:=p-e_y,
\]
where \(e_y\in\mathbb{R}^K\) is the \(y\)-th standard basis vector, one obtains
\[
\eta_j^{(\zeta)}=\frac{1}{\tau}a^\top\zeta^{(j)},
\qquad
\eta_j^{(\xi)}=\frac{1}{\tau^2}(w^{(j)})^\top\xi,
\]
and therefore, with
\[
\bar\mu_{\zeta,j}:=\EV[a^\top\zeta^{(j)}],
\qquad
\bar\mu_{\xi,j}:=\EV[(w^{(j)})^\top\xi],
\]
\[
\bar\sigma_{\zeta,j}^2:=\Var(a^\top\zeta^{(j)}),
\qquad
\bar\sigma_{\xi,j}^2:=\Var((w^{(j)})^\top\xi),
\]
assuming that \(\xi\) and \(\zeta^{(j)}\) are obtained from independent circuit evaluations, the first-order perturbation is approximately Gaussian with
\begin{equation}
\begin{aligned}
m
&=
\frac{1}{\tau}\bar\mu_{\zeta,j}
+
\frac{1}{\tau^2}\bar\mu_{\xi,j}\\
v
&=
\frac{1}{\tau^2}\bar\sigma_{\zeta,j}^2
+
\frac{1}{\tau^4}\bar\sigma_{\xi,j}^2
\label{eq:m_and_v}
\end{aligned}
\end{equation}

This proves the approximation \eqref{eq:gradient-plus-noise}.
\end{proof}

\clipbound*

\begin{proof}
Fix a data point $(x_i, y_i)$ and write $\rho_0 := \rho(x_i)$ for the encoded state. We begin with the logit gradient bound. Fix a class $k$, and consider first the single-parameter case $U(\theta)=e^{-i\theta H}$ for a Hermitian generator $H$. Differentiating $\rho(\theta)=U\rho_0 U^\dagger$ gives
\[
\partial_\theta \rho
=
(\partial_\theta U)\rho_0 U^\dagger + U\rho_0 (\partial_\theta U^\dagger)
=
-iHU\rho_0 U^\dagger + iU\rho_0 U^\dagger H
=
-i[H,\rho]
\]
Therefore, by cyclicity of the trace,
\[
\partial_\theta f_k
=
\Tr(O_k\,\partial_\theta\rho)
=
-i\,\Tr(O_k[H,\rho])
=
-i\,\Tr([O_k,H]\rho)
=
i\,\Tr(\rho[H,O_k]),
\]
and hence
\[
|\partial_\theta f_k|
=
\big|i\,\Tr(\rho[H,O_k])\big|
\;\le\;
\|\rho\|_1\,\|[H,O_k]\|
\;\le\;
2\,\|H\|\,\|O_k\|,
\]
using $\|\rho\|_1=1$. For the multi-parameter unitary, write
\[
U(\theta) = A_j\,e^{-i\theta_j H_j}\,V_j
\]
for each $j$, where $A_j$ collects the gates to the left of $e^{-i\theta_j H_j}$ and $V_j$ those to the right. Then
\[
\partial_{\theta_j}U
=
A_j(-iH_j e^{-i\theta_j H_j})V_j
=
-i(A_jH_jA_j^\dagger)\,U,
\]
similar to the single-parameter case, except with the evolved generator $H = A_j H_j A_j^\dagger$. Thus
\[
\partial_{\theta_j} f_k
=
i\,\Tr(\rho[H, O_k])
=
i\,\Tr\!\big(\tilde\rho_j [H_j,\,A_j^\dagger O_k A_j]\big),
\]
where $\tilde\rho_j = e^{-i\theta_j H_j}\,V_j\,\rho_0\,V_j^\dagger\,e^{i\theta_j H_j}$ by cyclicity. As before,
\[
|\partial_{\theta_j} f_k|
\;\le\;
\|[H, O_k]\|
=
\|[A_j H_j A_j^\dagger, O_k]\|
=
\|[H_j,\, A_j^\dagger O_k A_j]\|
\;\le\;
2\|H_j\|\,\|O_k\|
\]
Summing over coordinates yields
\begin{equation}
\|\partial_\theta f_k(x_i,\theta)\|_2^2
=
\sum_{j=1}^{p}|\partial_{\theta_j}f_k|^2
\;\le\;
4\|O_k\|^2\sum_{j=1}^{p}\|H_j\|^2
\;\le\;
G^2,
\label{eq:main-grad-bound}
\end{equation}
For the loss gradient, stack the per-coordinate identities~\eqref{eq:obj-to-loss-grad} for $j=1,\dots,p$ into a single vector equation. Viewing each $\partial_\theta f_k(x_i,\theta)$ as a vector in $\mathbb{R}^{p}$, the per-example loss gradient is
\[
g_i(\theta)
=
\partial_\theta\ell(f(x_i,\theta), y_i)
=
\frac{1}{\tau}\left(\sum_{k=1}^{K} p_k(x_i,\theta)\,\partial_\theta f_k(x_i,\theta) - \partial_\theta f_{y_i}(x_i,\theta)\right)
\]
Taking the Euclidean norm and applying the triangle inequality,
\[
\|g_i(\theta)\|_2
\;\le\;
\frac{1}{\tau}\left(\sum_{k=1}^{K} p_k(x_i,\theta)\,\|\partial_\theta f_k(x_i,\theta)\|_2 + \|\partial_\theta f_{y_i}(x_i,\theta)\|_2\right)
\;\le\;
\frac{2G}{\tau},
\]
where the last inequality uses $\sum_k p_k = 1$ together with the logit-gradient bound~\eqref{eq:main-grad-bound}, which holds uniformly over $k$.

It remains to bound the clipping bias. By definition $\bar g_i(\theta) - g_i(\theta) = \Pi_C(g_i(\theta)) - g_i(\theta)$, and for any vector $g$ one has $\|\Pi_C(g) - g\|_2 = (\|g\|_2 - C)_+$. Hence Jensen's inequality together with the uniform bound~\eqref{eq:loss-grad-bound} gives
\[
\|b_C(\theta)\|_2
=
\big\|\EV[\bar g_i(\theta) - g_i(\theta)]\big\|_2
\;\le\;
\EV\,\big(\|g_i(\theta)\|_2 - C\big)_+
\;\le\;
\left(\frac{2G}{\tau} - C\right)_+,
\]
which is~\eqref{eq:batch-clipping-bound}.
\end{proof}

\end{document}